\documentclass[11pt]{article}
\usepackage{fullpage}
\usepackage[latin2]{inputenc}
\usepackage[english]{babel}
\usepackage[cmex10]{amsmath}
\usepackage{todonotes}
\usepackage{amsthm}
\usepackage{amssymb}
\usepackage{microtype}
\usepackage{xspace}
\usepackage{todonotes}
\usepackage{comment}
\usepackage{enumerate}
\usepackage{tikz}

\usepackage{graphicx}
\usepackage{caption}
\usepackage{subcaption}

\numberwithin{equation}{section}
\newtheorem{theorem}{Theorem}[section]
\newtheorem{lemma}[theorem]{Lemma}
\newtheorem{claim}[theorem]{Claim}

\theoremstyle{definition}
\newtheorem{definition}[theorem]{Definition}

\newcommand{\ifshort}[2]{\ifthenelse{\equal{\useshort}{yes}}{#1}{#2}}

\newcommand{\onlyfull}[1]{\ifshort{}{#1}}
\newcommand{\useshort}{no}

\newcommand{\Oh}{\mathcal{O}}

\def\cqedsymbol{\ifmmode$\lrcorner$\else{\unskip\nobreak\hfil
\penalty50\hskip1em\null\nobreak\hfil$\lrcorner$
\parfillskip=0pt\finalhyphendemerits=0\endgraf}\fi} 

\newcommand{\cqed}{\renewcommand{\qed}{\cqedsymbol}}

\newcommand{\stepref}[1]{Claim~\ref{#1}}
\newenvironment{step}[1]{\begin{claim}\label{#1}}{\end{claim}}
\newenvironment{stepproof}{\begin{proof}}{\cqed\end{proof}}

\newcommand{\pmc}{\Omega}
\newcommand{\FF}{\mathcal{F}}

\newcommand{\cG}{\widehat{G}}

\newcommand{\sephit}{\alpha}
\newcommand{\pmchit}{\beta}
\newcommand{\nukhit}{\gamma}
\newcommand{\nukeps}{\eta}
\newcommand{\nukth}{\tau}
\newcommand{\meas}{\mu}

\newcommand{\ISname}{\textsc{Maximum Weight Independent Set}\xspace{}}
\newcommand{\EDname}{\textsc{Maximum Weight Efficient Dominating Set}\xspace{}}

\newcommand{\ISnameUW}{\textsc{Independent Set}\xspace{}}
\newcommand{\EDnameUW}{\textsc{Efficient Dominating Set}\xspace{}}

\title{Independence and Efficient Domination on $P_6$-free Graphs} 

\author{Daniel Lokshtanov\thanks{University of Bergen, Norway, \texttt{daniello@ii.uib.no}.} \and Marcin Pilipczuk\thanks{University of Warwick, UK, \texttt{malcin@mimuw.edu.pl}.} \and Erik Jan van Leeuwen\thanks{Max-Planck Institut f\"{u}r Informatik, Saarbr\"{u}cken, Germany, \texttt{erikjan@mpi-inf.mpg.de}.}}

\date{}

\begin{document}
\maketitle

\begin{abstract}
In the \ISname{} problem, the input is a graph $G$, every vertex has a non-negative integer weight, and the task is to find a set $S$ of pairwise non-adjacent vertices, maximizing the total weight of the vertices in $S$. We give an $n^{\Oh (\log^2 n)}$ time algorithm for this problem on graphs excluding the path $P_6$ on $6$ vertices as an induced subgraph.
Currently, there is no constant $k$ known for which \ISname{} on $P_{k}$-free graphs becomes NP-complete, and our result implies that if such a $k$ exists, then $k > 6$ unless all problems in NP can be decided in (quasi)polynomial time.

Using the combinatorial tools that we develop for the above algorithm, we also give a polynomial-time algorithm for \EDname{} on $P_6$-free graphs. In this problem, the input is a graph $G$, every vertex has an integer weight, and the objective is to find a set $S$ of maximum weight such that every vertex in $G$ has exactly one vertex in $S$ in its closed neighborhood, or to determine that no such set exists. Prior to our work, the class of $P_6$-free graphs was the only class of graphs defined by a single forbidden induced subgraph on which the computational complexity of \EDname{} was unknown.

\end{abstract}

\section{Introduction}\label{sec:intro}
An {\em independent set} in a graph $G$ is a set $S$ of pairwise non-adjacent vertices. In the \ISnameUW{} problem the input is a graph $G$ on $n$ vertices and an integer $t$, and the task is to determine whether $G$ contains an independent set of size at least $t$. \ISnameUW{} is a fundamental and extremely well-studied graph problem. It was one of the very first problems to be shown NP-complete~\cite{GJ79,Karp72}, and a significant amount of research~\cite{Alekseev04,BKKM99,Corneil1981,GLS81,LokshtanovVV14,LozinM06,Minty80,Sbihi1980}\footnote{This list is far from exhaustive, see the Information System on Graph Classes and their Inclusions (ISGCI)~\cite{graphclassesPage}.} has gone into identifying classes of graphs on which the problem becomes polynomial-time solvable.

A complete classification of the complexity status of \ISnameUW{} on all classes of graphs seems out of reach. However,  obtaining such a classification for all classes of graphs defined by excluding a single connected graph $H$ as an induced subgraph (we call such graphs $H$-free) looks like an attainable, yet very challenging, goal.  In particular, Alekseev~\cite{Alekseev82} showed in 1982 that \ISnameUW{} remains NP-complete on $H$-free graphs whenever $H$ is connected, but neither a path nor a subdivision of the claw. Since then, the complexity of \ISnameUW{} on classes of $P_k$-free graphs (we denote by $P_k$ the path on $k$ vertices) has been subject to intense scrutiny, but yielding rather modest progress. For $P_4$-free graphs a polynomial-time algorithm was given by Corneil et al.~\cite{Corneil1981} in 1981, and it took more than 30 years until a polynomial-time algorithm for the problem on $P_5$-free graphs was discovered by Lokshtanov et al.~\cite{LokshtanovVV14} in 2014. In the meanwhile, a substantial amount of work was devoted to \ISnameUW{} on subclasses of $P_5$-free graphs~\cite{BBKRS05,BL03,BM03,GeLo03,Mo97,zverovich04}, and some progress has been reported on subclasses of $P_6$-free graphs~\cite{Karthick15,Mo12,Mosca12B,Mosca13}.

In this paper we push the boundary of knowledge on the complexity of \ISnameUW{} on $P_k$-free graphs a step forward by giving a $n^{O(\log^2 n)}$-time algorithm for \ISnameUW{} on $P_6$-free graphs. Our algorithm also works for the weighted version of the problem. Here every vertex has a non-negative integer weight and we are looking for an independent set that maximizes the sum of the weights of the vertices in it.

\begin{theorem}\label{thm:main}
There is an $n^{\Oh(\log^2 n)}$-time, polynomial-space algorithm for \ISname{} on $P_6$-free graphs.
\end{theorem}

The algorithm of Theorem~\ref{thm:main} does not completely resolve the complexity status of \ISnameUW{} on $P_6$-free graphs, as it runs in quasipolynomial time rather than polynomial time. However, Theorem~\ref{thm:main} does imply that \ISname{} on $P_6$-free graphs is not NP-complete, unless all problems in NP can be solved in quasipolynomial time. This hints at the existence of a polynomial-time algorithm for the problem also on $P_6$-free graphs.

On the way to developing our algorithm for \ISnameUW{}, we prove several new combinatorial properties of $P_6$- and $P_7$-free graphs. We leverage these new combinatorial insights to develop a polynomial-time algorithm for  the \EDnameUW{} problem on $P_6$-free graphs. We say that a vertex {\em dominates} itself and all of its neighbors. An {\em efficient dominating set} in a graph $G$ is a vertex set $S$ such that every vertex $v$ in the graph is dominated by exactly one vertex in $S$. Not all graphs have an efficient dominating set, and in the \EDnameUW{} problem the input is a graph $G$ and the task is to determine whether $G$ has an efficient dominating set. We remark that the problem also goes by the name {\sc Perfect Code}~\cite{biggs1973perfect}. Observe that we do not ask for the smallest or largest efficient dominating set, only whether there exists one. This is because whenever a graph $G$ has an efficient dominating set, all such sets have the same cardinality~\cite{haynes1998fundamentals}. In the weighted variant, called \EDname{}, every vertex has an integer weight and the task is to find a maximum weight efficient dominating set, if one exists. Since the weights may be negative, there is no real difference between maximizing and minimizing the weight of the solution. Our second main theorem is the following.

\begin{theorem}\label{thm:effdom}
There is an $n^{\Oh (1)}$-time algorithm for \EDname{} on $P_6$-free graphs.
\end{theorem}

Prior to our work, the $P_6$ was the only graph $H$, connected or not, for which the complexity of {\sc Efficient Dominating Set} on $H$-free graphs was unknown~\cite{BrandstadtG14}. Thus our work completes the complexity classification of  \EDnameUW{} (and \EDname{}) on classes of graphs defined by a single forbidden induced subgraph and resolves the main open problem of~\cite{ab1,BrandstadtG14,ab2,BrandstadtMN13}. We remark that an alternative polynomial-time algorithm for \EDname{} has been independently obtained by Brandst\"{a}dt and Mosca~\cite{ab-effdom} using different methods.\footnote{Although~\cite{ab-effdom} appeared on arXiv a few weeks after this paper, the authors of~\cite{ab-effdom} contacted us and shared with us a preliminary version of~\cite{ab-effdom} immediately after we posted our work.}

\paragraph{Methodology.} 
The polynomial-time algorithm for \ISname{} on $P_5$-free graphs of Lokshtanov et al.~\cite{LokshtanovVV14} demonstrated that investigating {\em potential maximal cliques} and {\em minimal separators} (see Section~\ref{sec:prelims} for definitions) yields valuable insights on the structure of $P_5$-free graphs. Our algorithm for $P_6$-free graphs is also based on studying potential maximal cliques and minimal separators. However, this is where the similarity between the two algorithms ends, as essentially all of the arguments used in the algorithm for $P_5$-free graphs quickly break down for $P_6$-free graphs.

At heart our algorithm is very simple: the algorithm picks a node $v$ and proceeds recursively in two branches. In the first $v$ is included in the independent set, and the algorithm needs to solve $G - N(v)$ recursively. In the second $v$ is excluded from the independent set, and the algorithm is called recursively on $G - v$. If in any recursive call the graph becomes disconnected the algorithm solves the connected components independently. The crux of the analysis is to show that one can always cleverly chose the vertex $v$, such that after only a few branches either the size of the graph decreases by at least $.1n$, or the graph breaks into connected components of size at most $.9n$. 

Roughly speaking, the vertex $v$ to branch on is chosen as follows. The algorithm identifies a {\em nuke} in $G$: a relatively small vertex set $S$ such that every connected component of $G - X$ has size at most $.9n$ (for a formal definition of a nuke, see Definition~\ref{def:nuke}). The algorithm then picks a vertex $v$ with a large neighborhood in $S$ to branch on. In order to guarantee the existence of a nuke $S$ and a vertex $v$ with a large neighborhood in $S$ we prove the following theorem about minimal separators in $P_7$-free graphs. 

\begin{theorem}\label{thm:hit-sep}
There exists a positive constant $\sephit > 0$ such that for every $P_7$-free graph $G$, for every minimal separator $S$ in $G$, and for every probability measure $\meas$ on $S$, there exists a vertex $v \in V(G)$ satisfying $\meas(N(v)) \geq \sephit$.
\end{theorem}

The reason that Theorem~\ref{thm:hit-sep} is not already sufficient to yield a quasipolynomial-time algorithm for \ISnameUW{} on $P_7$-free graphs is that, despite the similarity between the definitions of nukes and minimal separators, not all nukes are minimal separators. For $P_6$-free graphs we are able to prove an analogue of Theorem~\ref{thm:hit-sep} for nukes rather than minimal separators, and this is sufficient to give a $n^{\Oh (\log^2 n)}$-time algorithm for \ISname{}. As a first step to lift Theorem~\ref{thm:hit-sep} to work for nukes we generalize it to potential maximal cliques in $P_7$-free graphs.

\begin{theorem}\label{thm:hit-pmc}
There exists a positive constant $\pmchit > 0$ such that for every connected $P_7$-free graph $G$ on at least two vertices,
for every potential maximal clique $\pmc$ in $G$, and every probability measure $\meas$ on $\pmc$,
there exists a vertex $v \in V(G)$ satisfying $\meas(N(v)) \geq \pmchit$.
\end{theorem}

Theorem~\ref{thm:hit-pmc} turns out to be very useful not only in our quasipolynomial-time algorithm for \ISname{}, but for the polynomial-time algorithm for \EDname{} as well. Indeed, an almost immediate consequence of Theorem~\ref{thm:hit-pmc} is that for any $P_7$-free graph $G$, any efficient dominating set $X$ in $G$ and any potential maximal clique $\Omega$ in $G$, $|X \cap \Omega| \leq 1/\pmchit$ (see Lemma~\ref{lem:ed:in-pmc} for a simple proof).

The observation above strongly suggests that one can solve  \EDname{} on $P_7$-free graphs in polynomial time by doing dynamic programming over the tree decomposition of an arbitrarily chosen minimal triangulation of $G$. For $P_7$-free graphs this approach fails, as is expected from the NP-completeness of {\sc Efficient Dominating Set} on $P_7$-free graphs~\cite{BrandstadtMN13,YenL96}. On the other hand, for $P_6$-free graphs, we are able to carry this approach through.

We mention here that this approach follows a completely disjoint direction from the one followed in recent papers~\cite{ab1,ab2} that gave polynomial-time algorithms for subclasses of $P_{6}$-free graphs.
In particular, those papers show that one can reduce to \ISname{} on the square of the graph by proving special properties of the square when the graph is from such a subclass and has an efficient dominating set.

\paragraph{Outline of the paper.} In Section~\ref{sec:prelims} we set up the definitions and necessary notations. In Section~\ref{sec:hitSep} we prove Theorems~\ref{thm:hit-sep} and~\ref{thm:hit-pmc}, while Section~\ref{sec:nukes} contains the generalization of Theorem~\ref{thm:hit-sep} to nukes. We then proceed to the main algorithmic results, Section~\ref{sec:ISalg} contains the quasipolynomial-time algorithm for \ISname{}, while Section~\ref{sec:EDalg} contains the polynomial-time algorithm for \EDname{}, both on $P_6$-free graphs. In Section~\ref{sec:conclusion} we conclude with some open problems and counterexamples to the most natural generalizations of the structural results underlying our algorithms.







\section{Preliminaries}\label{sec:prelims}
For all graph terminology not defined here, we refer to the monograph by Diestel~\cite{diestel}.
For a graph $G$ and sets $A,B \subseteq V(G)$, we denote $N_B(A) := N(A) \cap B$.

Let $G$ be a graph; throughout, we assume that all graphs are finite, simple, and undirected. Given distinct $s,t \in V(G)$, a set $S \subseteq V(G)$ is an \emph{$s$-$t$ separator} if $s$ and $t$ are in distinct components of $G \setminus S$. We say that $S \subseteq V(G)$ is a \emph{minimal} $s$-$t$ separator if no $S' \subsetneq S$ is also an $s$-$t$ separator. Then $S \subseteq V(G)$ is a \emph{(minimal) separator} of $G$ if $S$ is a (minimal) $s$-$t$ separator for some $s,t \in V(G)$. Given a separator $S \subseteq V(G)$, a component $C$ of $G\setminus S$ is said to be \emph{full} if every vertex of $S$ has a neighbor in $C$. It can be shown that $S$ is a minimal separator if and only if $G\setminus S$ has at least two full components.

A set $M \subseteq V(G)$ is a \emph{module} of $G$ if every vertex $v \in V(G) \setminus M$ is either fully adjacent or fully anti-adjacent to $M$; that is, either $vu \in E(G)$ for each $u \in M$ or $vu \not\in E(G)$ for each $u \in M$. 
A module $M$ of $G$ is \emph{trivial} if $M = V(G)$, $M = \emptyset$, or $|M| = 1$. A graph is \emph{prime} if it only has trivial modules. 
A \emph{modular partition} $\mathcal{M}$ of $G$ is a set of disjoint modules of $G$ with union $V(G)$. The \emph{quotient graph} $G / \mathcal{M}$ induced by $\mathcal{M}$ has a vertex for each module of $\mathcal{M}$ and has an edge between two vertices if the corresponding modules are fully adjacent to each other. Observe that, by definition, a non-edge between two vertices in the quotient graph implies that the corresponding modules are fully anti-adjacent to each other. 
A module $M$ of $G$ is \emph{proper} if $M \not= V(G)$. A module $M$ of $G$ is \emph{strong} if for every other module $M'$ of $G$, either $M \subseteq M'$, $M' \subseteq M$ or $M \cap M' = \emptyset$.

\begin{theorem}[{\cite[Theorem~2]{habibpaul}}] \label{thm:modular}
Let $G$ be a connected graph on at least two vertices and let $\mathcal{M}$ denote the set of maximal proper strong modules of $G$. Then $\mathcal{M}$ is a modular partition of $G$, and the quotient graph $G / \mathcal{M}$ is either a clique or a prime graph.
\end{theorem}

A graph $G$ is \emph{chordal} if every induced cycle of $G$ has length three. A \emph{triangulation} of a graph $G$ is a set $F \subseteq (V(G) \times V(G)) \setminus E(G)$ such that the graph $G+F := (V(G), E(G) \cup F)$ is chordal. We say that $F$ is a \emph{minimal} triangulation of $G$ if no $F' \subsetneq F$ is a triangulation of $G$. Then a \emph{potential maximal clique} of $G$ is a set $\pmc \subseteq V(G)$ such that $\pmc$ induces a maximal clique in some minimal triangulation of $G$. We need the following properties of potential maximal cliques due to Bouchitt{\'{e}} and Todinca~\cite{todinca}.

\begin{theorem}[\cite{todinca}]\label{thm:pmc-seps}
Let $G$ be a graph. If $\pmc \subseteq V(G)$ is a potential maximal clique of $G$, then
for every connected component $C$ of $G \setminus \pmc$, the set $N_G(C) \subseteq \pmc$ is a minimal separator of $G$.
\end{theorem}

\begin{theorem}[\cite{todinca}]\label{thm:pmc}
Let $G$ be a graph. A set $\pmc \subseteq V(G)$ is a potential maximal clique of $G$ if and only if
the following two conditions hold:
\begin{enumerate}
\item for every connected component $C$ of $G \setminus \pmc$, we have $N_G(C) \subsetneq \pmc$;
\item for every two distinct vertices $x,y \in \pmc$, either $xy \in E(G)$
or there exists a component $C$ of $G \setminus \pmc$ such that $x,y \in N_G(C)$
(in this case we say that the non-edge $xy$ is \emph{covered} by the component $C$). 
\end{enumerate}
\end{theorem}

\section{Hitting Separators and Potential Maximal Cliques}\label{sec:hitSep}

\subsection{Proof of Theorem~\ref{thm:hit-sep}}
Let $G$ be a graph, let $S$ be a minimal separator of $G$, let $\mu$ be any probability measure on $S$, and let $\sephit < 1$ be some constant chosen later. 
For sake of contradiction, assume that for every $v \in V(G)$ we have $\mu(N(v)) < \sephit$. 
This implies that $\mu(x) < \sephit$ for every $x \in S$, because by the minimality of $S$, $x$ has a neighbour $v$ in some (full) component of $G \setminus S$, and thus $\mu(x) \leq \mu(N(v)) < \sephit$.

Let $A$ and $B$ be two full components of $G \setminus S$.
We say that $x \in S$ is \emph{lucky} (with respect to $A$) if there exists an induced $P_4$ in $G$
with one endpoint in $x$ and the remaining three vertices in $A$.
We say that a pair $(x,y) \in S \times S$ is \emph{lucky} (with respect to $A$) if $x$ is lucky or there exists an induced $P_4$ in $G$ with endpoints $x$ and $y$ and its middle two vertices in $A$.
The following lemma is the crucial step in the argumentation.

\onlyfull{
\begin{figure}[tb]
\begin{center}
\includegraphics{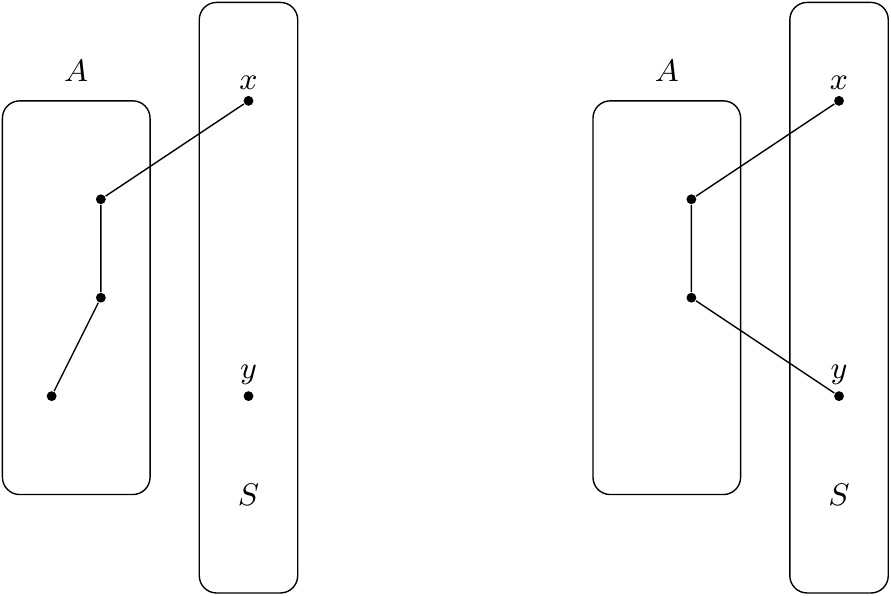}
\caption{Two possibilities for $(x,y)$ being lucky.}
\label{fig:lucky}
\end{center}
\end{figure}
}

\begin{lemma}\label{lem:lucky}
Let $G$ be a graph, and let $S$, $\mu$, $\sephit$, $A$, and $B$ as above.
If we choose two vertices $x,y \in S$ independently at random according to distribution $\mu$,
then the probability that $(x,y)$ is not lucky with respect to $A$ (or $B$) is less than $6\sephit$.
\end{lemma}
\begin{proof}
If $|A| = 1$, then the single vertex $a$ of $A$ is adjacent to all vertices of $S$, as $A$ is a full component of $G \setminus S$. Hence, $\mu(N(a)) = 1 > \sephit$, a contradiction. Therefore, $|A| > 1$.

Consider the graph $G[A]$, and let $\mathcal{M}$ be the family of maximal proper strong modules of $G[A]$. Note that $\mathcal{M}$ is a modular partition and that the quotient graph of this partition is a clique or a prime graph
\ifshort{(see e.g.\ \cite[Theorem~2]{habibpaul})}{(Theorem~\ref{thm:modular})}, since $G[A]$ is connected and $|A| > 1$. 
Now pick two arbitrary vertices $p,q \in A$ in two distinct modules of $\mathcal{M}$ that are adjacent in the quotient graph. We can indeed pick such $p,q$, because the quotient graph is connected, as $G[A]$ is connected. Moreover, $|\mathcal{M}| > 1$, since the set of singleton modules (one for each vertex) is a family of at least two proper strong modules (recall that $|A| > 1$). 

Consider some $(x,y) \in S \times S$ that are chosen independently at random according to distribution $\mu$. In the following, we continuously use that $\mu(N(v)) < \sephit$ for every $v \in V(G)$ and $\mu(u) < \sephit$ for every $u \in S$. 
With probability less than $2\sephit$ we have $x \in N(p) \cup N(q)$, and with probability less than $2\sephit$ we have $x = y$ or $xy \in E(G)$.
Furthermore, with probability less than $\sephit$ we have $N(y) \cap A \subseteq N(x) \cap A$, since for a fixed choice of $y$ and $v \in N(y) \cap A$,
the probability that $x \in N(v)$ is at most $\sephit$. 
Now assume that none of the aforementioned events happen, and pick arbitrary $r \in (N(y) \setminus N(x)) \cap A$.

Let $\mathcal{C}$ be the family of connected components of $G[A \setminus N(x)]$. Consider any $C \in \mathcal{C}$ and any vertex $v \in N(x) \cap A$. If $v$ is neither fully adjacent nor fully anti-adjacent to $C$, then since $C$ is connected, there exist two neighbouring vertices $u,w \in C$ such that $u \in N(v)$ and $w \not\in N(v)$. Since $u,w \not\in N(x)$ by the definition of $C$ and $\mathcal{C}$, the vertices $x,v,u,w$ form a $P_4$ in $G$ with one endpoint in $x$; then, $x$ and by extension $(x,y)$ is lucky. Hence, we may assume that for every $C \in \mathcal{C}$ and every $v \in N(x) \cap A$, the vertex $v$ is either fully adjacent or fully anti-adjacent to $C$. In particular, every $C \in \mathcal{C}$ is a module of $G[A]$. 


Consider the component $C \in \mathcal{C}$ that contains the vertex $p$; note that $C$ exists, because  $x \not\in N(p)$ by assumption\onlyfull{ (see Fig.~\ref{fig:hitsep})}. Since $C$ is a module of $G[A]$ and $\mathcal{M}$ is the family of maximal proper strong modules of $G[A]$,
either there exists a module $M \in \mathcal{M}$ that contains $C$, or $C$ is a union of several modules of $\mathcal{M}$ and the quotient graph $G[A] / \mathcal{M}$ is a clique. 

If $C \subseteq M$ for some $M \in \mathcal{M}$, then consider the module $M' \in \mathcal{M}$ that contains $q$. By the choice of $M$ and $M'$, $M'$ is fully adjacent to $M$, and in particular, $M'$ is fully adjacent to $C$.
Since $C \in \mathcal{C}$ and $C \subseteq M$, we have that $M'$ cannot contain any vertices of any other component in $\mathcal{C}$. Hence, $M' \subseteq N(x)$. However, $q \notin N(x)$, a contradiction.

Therefore, $C$ is a union of several modules of $\mathcal{M}$ and the quotient graph $G[A] / \mathcal{M}$ is a clique. Then $C$ is fully adjacent to $A \setminus C$, and in particular to every $C' \in \mathcal{C} \setminus \{C\}$, which implies that $\mathcal{C} = \{C\}$. 
Therefore, there exists a vertex $v \in N(x) \cap A$ with $A \setminus N(x) \subseteq N(v)$, because $G[A]$ is connected and $C = A \setminus N(x)$ is a module. 
Observe that $y \in N(v)$ with probability less than $\sephit$, since $\mu(N(v)) < \sephit$. If this does not happen (i.e., $y \not\in N(v)$), then $x,v,r,y$ form a $P_4$, because $r \in (N(y) \setminus N(x)) \cap A$. Hence, $(x,y)$ is lucky.

By the union bound, the total probability that any of the aforementioned events happen is at most $6\sephit$. The lemma follows.
\end{proof}

\onlyfull{
\begin{figure}[tb]
\begin{center}
\includegraphics{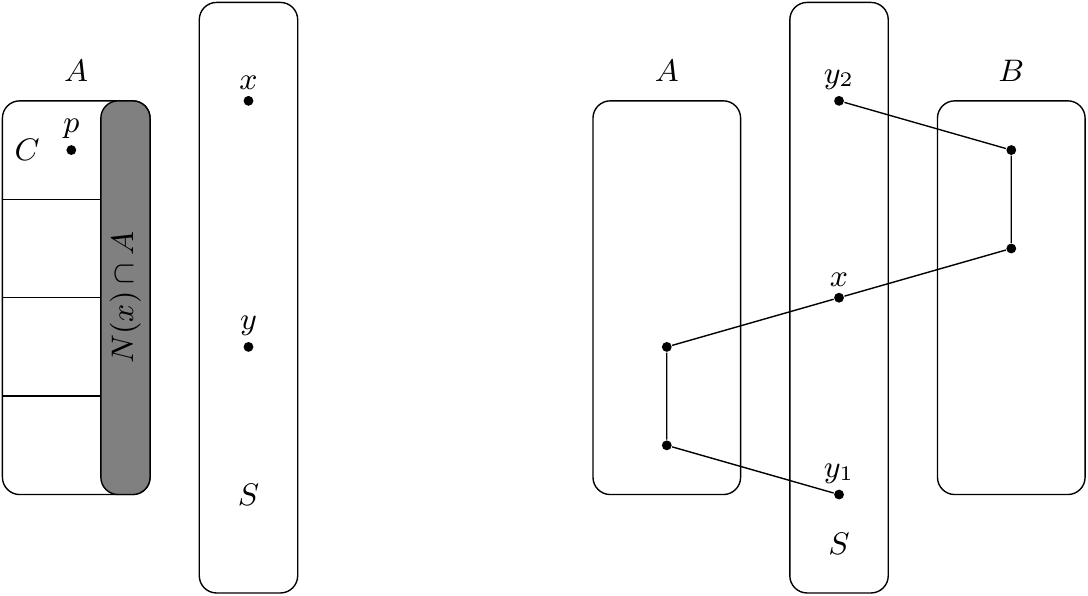}
\caption{The left panel shows part of the reasoning of Lemma~\ref{lem:lucky} with the choice of $C$ being the connected component of $G[A \setminus N(x)]$ that contains $p$.
  The right panel shows an archetypical $P_7$ constructed in the proof of Theorem~\ref{thm:hit-sep}.}
\label{fig:hitsep}
\end{center}
\end{figure}
}

We are now ready to conclude the proof of Theorem~\ref{thm:hit-sep}.
Pick three vertices $x,y_1,y_2 \in S$ independently at random, each with distribution $\mu$. The goal will be to find a $P_4$ in $A$ from $x$ (possibly to $y_1$) and a $P_4$ in $B$ from $x$ (possibly to $y_2$) that jointly form a $P_7$ in $G$. 
Consider the following set of ``bad'' events. In the below, we repeatedly rely on Lemma~\ref{lem:lucky} and the assumptions that $\mu(N(v)) < \sephit$ for every $v \in V(G)$ and $\mu(u) < \sephit$ for every $u \in S$.
\begin{itemize}
\item $(x,y_1)$ is not lucky with respect to $A$; this happens with probability less than $6\sephit$. Otherwise, let $P^1$ be the witnessing $P_4$.
\item $(x,y_2)$ is not lucky with respect to $B$; this happens with probability less than $6\sephit$. Otherwise, let $P^2$ be the witnessing $P_4$.
\item Some vertices from the set $\{x,y_1,y_2\}$ are equal or adjacent; this happens with probability less than $6\sephit$.
\item One of the (two or three) vertices from $V(P^1) \cap A$ is adjacent to $y_2$; since the choice of $x$ and $y_1$ is independent of the choice of $y_2$, and the path $P^1$ is a function of the pair $(x,y_1)$ only,
  this happens with probability less $3\sephit$ ($y_2$ needs to land outside the neighbourhoods of $V(P^1) \cap A$).
\item One of the (two or three) vertices from $V(P^2) \cap B$ is adjacent to $y_1$; since the choice of $x$ and $y_2$ is independent of the choice of $y_1$, and the path $P^2$ is a function of the pair $(x,y_2)$ only,
  this happens with probability less than $3\sephit$ ($y_1$ needs to land outside the neighbourhoods of $V(P^2) \cap B$).
\end{itemize}
By the union bound, the probability that none of the aforementioned ``bad'' events happen is greater than $1-24\sephit$. Hence, for sure when $\sephit = \frac{1}{24}$, there is a choice of $x,y_1,y_2 \in S$ for which the paths $P^{1}$ and $P^{2}$ exist and jointly form a $P_7$ in $G$. Hence, if $G$ is $P_7$-free, then there is a vertex $v \in V(G)$ satisfying $\mu(N(v)) \geq \sephit$ for some constant $\sephit > 0$ (in fact even $\sephit \geq \frac{1}{24}$).

\subsection{Proof of Theorem~\ref{thm:hit-pmc}}
The main tool is the following general lemma.

\begin{lemma}\label{lem:pmc-main}
Let $H$ be a graph on $n_H$ vertices and $m_H$ edges,
let $G$ be a graph, let $\pmc$ be a potential maximal clique in $G$, and let $\mu$ be a probability measure on $\pmc$. Then there exists either:
\begin{enumerate}
\item a vertex $v \in V(G)$ with $\mu(v) > \frac{1}{2n_H^2}$ or with $\mu(N(v)) > \frac{1}{2n_H^2}$;
\item a minimal separator $S \subseteq V(G)$ with $\mu(S) > \frac{1}{2n_Hm_H}$; or
\item an induced subgraph of $G$ isomorphic to a graph obtained from $H$ by replacing every edge by a path of length at least two (i.e., subdividing at least once).
\end{enumerate}
\end{lemma}
\begin{proof}
Let $H$, $G$, $\pmc$, and $\mu$ be as in the statement, and assume for sake of contradiction that neither of the first two outcomes happen.
Consider the following random experiment: independently, for every $p \in V(H)$,
choose a vertex $x_p \in \pmc$ according to the distribution $\mu$.

For two distinct vertices $p,q \in V(H)$, we have $x_p = x_q$ with probability at most $\frac{1}{2n_H^2}$, and $x_px_q \in E(G)$ (i.e., $x_q \in N(x_p)$)
  with probability at most $\frac{1}{2n_H^2}$. Consequently, all vertices $X := \{x_p : p \in V(H)\}$ are pairwise distinct and nonadjacent with probability
  at least
  $$1 - \binom{n_H}{2} \cdot 2 \cdot \frac{1}{2n_H^2} > \frac{1}{2}.$$

Assume that the aforementioned event happens.
For every two distinct and nonadjacent vertices $x,y \in \pmc$, fix a component $C(x,y)$ of $G \setminus \pmc$ that covers the non-edge $xy$ (i.e., $x,y \in N(C)$).
For every edge $pq \in E(H)$, consider the component $C_{pq} := C(x_p,x_q)$. As the choices of $x_r$ for distinct vertices $r \in V(H)$ are independent,
the probability that $x_r \in N(C_{pq})$ for a fixed $r \in V(H) \setminus \{p,q\}$ is at most $\frac{1}{2n_Hm_H}$, since $N(C_{pq})$ is a minimal separator of $G$ by Theorem~\ref{thm:pmc-seps} and thus $\mu(N(C_{pq})) \leq \frac{1}{2n_Hm_H}$ by assumption. 
Consequently, the probability that $X$ is an independent set of size $h$ and for every $pq \in E(H)$ we have $N(C_{pq}) \cap X = \{x_p,x_q\}$ is strictly greater than
$$\frac{1}{2} - n_H m_H \cdot \frac{1}{2n_H m_H} = 0.$$

If this event happens, then for every $pq \in E(H)$ choose a shortest path between $x_p$ and $x_q$ with internal vertices in $C_{pq}$. 
The union of all aforementioned paths forms an induced subgraph of $G$ isomorphic to a graph obtained from $H$ by replacing each edge with a path of length at least two, obtaining the last outcome.
\end{proof}

\onlyfull{
\begin{figure}[tb]
\begin{center}
\includegraphics{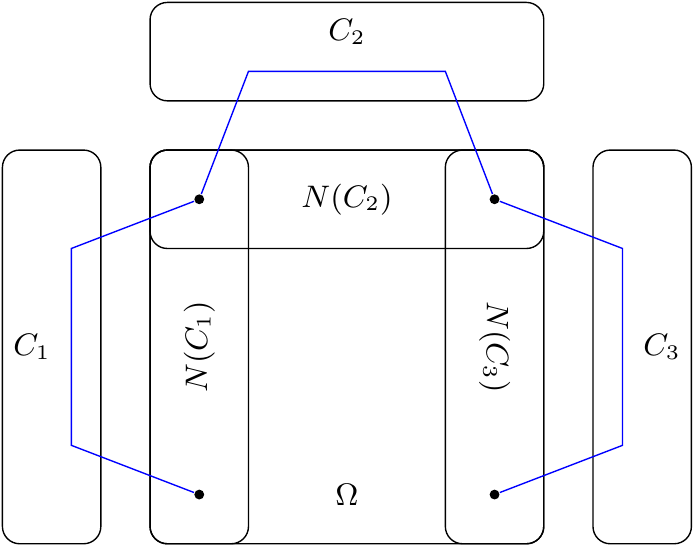}
\caption{Illustration of the proof of Lemma~\ref{lem:pmc-main} for $H=P_4$.}
\label{fig:hitpmc}
\end{center}
\end{figure}
}

We now prove Theorem~\ref{thm:hit-pmc} using Lemma~\ref{lem:pmc-main}. 
Let $G$ be a connected $P_7$-free graph on at least two vertices, let $\pmc$ be a potential maximal clique of $G$, let $\mu$ be any probability measure on $\pmc$, and let $\pmchit$ be some constant chosen later. Let $\alpha$ denote the constant of Theorem~\ref{thm:hit-sep}. For sure, if $\pmchit = \min\{\frac{\sephit}{24},\frac{1}{32}\}$, then the following happens. Apply Lemma~\ref{lem:pmc-main} with $H = P_4$ and consider its outcomes.
\begin{enumerate}
\item If $\mu(N(v)) > \frac{1}{32} \geq \pmchit$, then we are done. Otherwise, if $\mu(v) > \frac{1}{32} > \pmchit$, then by connectivity of $G$ there is a vertex $u \in N(v)$ with $\mu(N(u)) \geq \pmchit$.
\item Note that $\mu(S) > \frac{1}{24}$. Apply Theorem~\ref{thm:hit-sep} to $S$ and the restriction $\mu'$ of $\mu$ to $S$. It follows that there is a $v \in V(G)$ with $\mu'(N(v)) \geq \sephit$ and thus $\mu(N(v)) > \frac{\sephit}{24} \geq \pmchit$.
\item By the choice of $H$, this implies the existence of an induced $P_7$ in $G$, a contradiction.
\end{enumerate}
Therefore, there is a vertex $v \in V(G)$ satisfying $\mu(N(v)) \geq \pmchit$ for some contant $\pmchit > 0$ (in fact even $\pmchit \geq \frac{1}{576}$).


\section{Nuking a Graph}\label{sec:nukes}
In this section we study the following notion.

\begin{definition}[nuke, shelter]\label{def:nuke}
For a constant $0 < \nukeps \leq 0.1$ and a threshold $\nukth \geq 0$, a set of vertices $X$ is a \emph{$(\nukeps,\nukth)$-nuke} in a graph $G$
if the following holds:
\begin{itemize}
\item $(1-2\nukeps) |V(G)| \leq \nukth \leq (1-\nukeps)|V(G)|$
\item $|X| \leq \nukeps |V(G)|$;
\item for every connected component $C$ of $G-X$ we have $|C| + |X| \leq \nukth$.
\end{itemize}
Given a $(\nukeps,\nukth)$-nuke $X$ in $G$, any connected component of $G-X$ is called a \emph{shelter}.\footnote{The main motivation for introducing the notion of a shelter is to explicitly distinguish connected components of $G-X$ from connected components of $G-\pmc$ for some potential maximal clique $\pmc$; we will call the former \emph{shelters}, while the latter will be simply \emph{connected components}.}
\end{definition}
If the parameters $\nukeps$ and $\nukth$ are clear from the context, we will simply call the set $X$ \emph{a nuke in $G$}.

Intuitively, a nuke is a small set of vertices in $G$ whose removal breaks $G$ into connected components of multiplicatively smaller size.
Our algorithm keeps track of a nuke $X$ in the given input $P_6$-free graph $G$ and tries to branch on vertices of $G$ so that $X$ will be removed
from $G$ as quickly as possible. 
This motivation introduces two delicate issues that result in a slightly technical definition of a nuke.
First, during branching we need to keep the threshold $\nukth$ constant, while the size of $G$ drops a bit --- if we define 
the nuke so that, say, $|C| + |X| \leq (1-\nukeps) |V(G)|$, a set $X$ may stop to be a nuke due to a removal of a vertex from $G$ and consequent decrease of the bound $(1-\nukeps)|V(G)|$.
Second, we would like to argue about \emph{inclusion-wise minimal nukes}, which makes the measure $|C|+|X|$ (as opposed to simply $|C|)$ more natural, as we can then assume that every element of an inclusion-wise minimal nuke is adjacent to at least two shelters.

The rest of this section is devoted to a proof of the following structural statement.

\begin{theorem}\label{thm:hit-nuke}
There exists a constant $\nukhit > 0$ such that for every constant $0 < \nukeps \leq 0.1$, for every connected $P_6$-free graph $G$ on at least two vertices,
for every threshold $(1-2\nukeps)|V(G)| \leq \nukth \leq (1-\nukeps)|V(G)|$, for every inclusion-wise minimial $(\nukeps,\nukth)$-nuke $X$ in $G$, and for every probability measure $\meas$ on $X$, there exists a vertex $v \in V(G)$
with $\meas(N(v)) \geq \nukhit$.
\end{theorem}

Let $\nukeps$, $\nukth$, $G$, $X$, and $\meas$ be as in the statement of Theorem~\ref{thm:hit-nuke}.
We set $\nukhit = 0.1\pmchit \leq 0.1$, where the constant $\pmchit$ comes from Theorem~\ref{thm:hit-pmc}.
We will prove Theorem~\ref{thm:hit-nuke} by contradiction: assume that for every $v \in V(G)$ we have $\meas(N(v)) < \nukhit$.
We will unravel subsequent observations about the structure of $G$, leading to a final contradiction.

We start with the following observation.
\begin{step}{n:cmpl}
There exists a minimal triangulation $\cG$ of $G$, such that $X$ is a $(\nukeps,\nukth)$-nuke of $\cG$ as well and, moreover,
the shelters of $\cG-X$ are exactly the same as of $G-X$.
\end{step}
\begin{stepproof}
Consider the following completion $G_0$ of $G$: we first turn $X$ into a clique and then, for every shelter $C$ of $G-X$ we turn $C$ into a clique and make it completely adjacent to $X$.
Clearly, $G_0$ is a chordal graph, and the set of connected components of $G_0-X$ and $G-X$ are the same. Consequently, any minimal triangulation $\cG$ of $G$ that is a subgraph of $G_0$
has the required properties.
\end{stepproof}

We fix a minimal triangulation $\cG$ of $G$ satisfying the statement of~\stepref{n:cmpl}. Observe the following.

\begin{step}{n:pmc-shelter}
If two vertices $v,u \in V(G) \setminus X$ appear in the same maximal clique of $\cG$, then they are contained in the same shelter of $G-X$.
\end{step}
\begin{stepproof}
Recall that $X$ is a nuke of $\cG$ as well, with $\cG-X$ having the same set of shelters as $G-X$. Furthermore, $uv \in E(\cG)$.
\end{stepproof}

\begin{step}{n:small-pmc}
For any maximal clique $\pmc$ in $\cG$ we have $\meas(\pmc) \leq 0.1$.
\end{step}
\begin{stepproof}
If for some maximal clique $\pmc$ we have $\meas(\pmc) > 0.1$, then we are done by applying Theorem~\ref{thm:hit-pmc} to $\pmc$ and $\meas$ conditioned on $\pmc$.
\end{stepproof}

By standard arguments, there exists a maximal clique $\pmc$ in $\cG$ such that for every connected component $C$ of $G-\pmc$ we have $\meas(C) \leq 0.5$.
Fix one such maximal clique $\pmc$. We say that a component $C$ of $G-\pmc$ is \emph{nuked} if $C \cap X \neq \emptyset$.

\begin{step}{n:two-nuked}
There are at least two nuked components.
\end{step}
\begin{stepproof}
By the choice of $\pmc$, every nuked component contains at most half of the measure of $X$.
Furthermore, by~\stepref{n:small-pmc}, $\meas(\pmc) \leq 0.1$. Thus, there are at least two nuked components.
\end{stepproof}

By~\stepref{n:pmc-shelter}, all vertices of $\pmc \setminus X$ are contained in one shelter of $G-X$.
Let $D$ be this shelter; we set $D=\emptyset$ if $\pmc \subseteq X$.

\begin{step}{n:large-D}
$|D| \geq (0.5-3\nukeps)|V(G)| \geq 0.2|V(G)|$.
\end{step}
\begin{stepproof}
By~\stepref{n:two-nuked}, there exists a nuked component $C$ with $|C| \leq |V(G)|/2$.
Consider the set $X' = X \setminus C$. By the minimality of $X$, $X'$ is not a $(\nukeps,\nukth)$-nuke in $G$. As $|X'| < |X| \leq \nukeps |V(G)|$, the only reason for $X'$ to not be a nuke is that
there exists a shelter $C'$ of $G-X'$ that is too large, that is, $|C'| + |X'| > \nukth$.
By the construction of $X'$, the shelters of $G-X'$ and $G-X$ are the same, except for $C \cup D$, which is a shelter of $G-X'$, but may contain multiple shelters
of $G-X$. Therefore $C' = C \cup D$. Hence,
$$(1-2\nukeps)|V(G)| \leq \nukth < |C'| + |X'| \leq |C| + |D| + |X'| \leq |V(G)|/2 + |D| + |X| \leq (0.5 + \nukeps)|V(G)| + |D|.$$
\end{stepproof}
Note that~\stepref{n:large-D} in particular implies that $D \neq \emptyset$, that is, $\pmc$ is not completely contained in $X$.

\begin{step}{n:nei-D}
$X = N(D)$.
\end{step}
\begin{stepproof}
Clearly, $N(D) \subseteq X$. By the minimality of $X$, it suffices to show that $N(D)$ is a nuke in $G$.
Consider a shelter $D'$ of $G-N(D)$. If $D' = D$, then $|D'| + |N(D)| \leq |D| + |X| \leq \nukth$ by the assumption that $X$ is a nuke.
Otherwise, by~\stepref{n:large-D} and the assumption $\nukeps \leq 0.1$ we have
$$|D'| + |N(D)| \leq |V(G) \setminus D| \leq (0.5 + 3\nukeps)|V(G)| \leq (1-2\nukeps)|V(G)| \leq \nukth.$$
\end{stepproof}

\begin{step}{n:nuked-nei}
For every nuked component $C$ of $G-\pmc$ it holds that $N(C) \setminus X \neq \emptyset$, that is, there exists a non-nuked vertex in the neighbourhood of $C$.
\end{step}
\begin{stepproof}
A direct corollary from the facts that $X = N(D)$, $D$ is connected, and contains vertices of $\pmc$.
\end{stepproof}

\begin{step}{n:X-nei}
For every $x \in X$ there exists a shelter $D'$ of $G-X$ that is different than $D$ and contains a vertex adjacent to $x$.
\end{step}
\begin{stepproof}
If that is not the case, then $X \setminus \{x\}$ is a nuke as well, contradicting the minimality of $X$;
note that here we rely on the fact that we measure $|C| + |X|$ instead of just $|C|$ in the last property in the definition of a nuke.
\end{stepproof}

Our goal is now to exhibit a restricted structure of the nuked components of $G-\pmc$, using the fact that $G$ is $P_6$-free.
Intuitively, every nuked component gives rise to a potential $P_3$ or even $P_4$ sticking into such a component; by combining two such paths we should obtain a forbidden $P_6$.
The next four observations assert the existence of such sticking out $P_3$s and $P_4$s.

\begin{step}{n:P3}
For every nuked component $C$ of $G-\pmc$, and every $v \in N(C) \setminus X$, there exists a $P_3$ in $G$ with one endpoint in $v$ and the remaining two vertices in $C$.
\end{step}
\begin{stepproof}
See Fig.~\ref{fig:P3A} for an illustration of the proof.
Let $x \in C \cap X$, and let $D'$ be a shelter different from $D$ and adjacent to $x$, whose existence is asserted by~\stepref{n:nuked-nei}.
Since $N(C) \setminus X \subseteq D$, we have $D' \subseteq C$.
Consequently, $D' \cap N(v) = \emptyset$, in particular $C$ is not contained in $N(v)$. The existence of the asserted $P_3$ follows from the connectivity of $C$.
\end{stepproof}

\begin{figure}[tb]
\begin{center}
\includegraphics{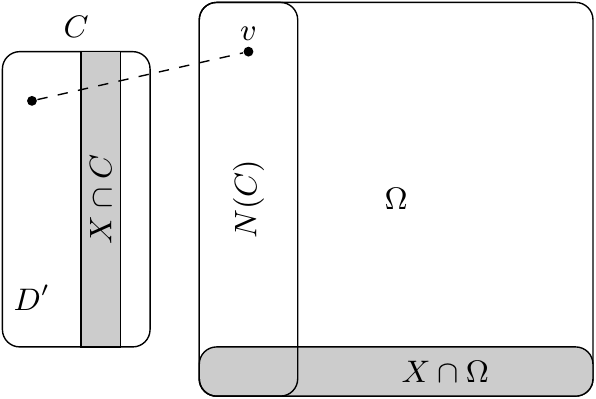}
\caption{Illustration of the proof of Claim~\ref{n:P3}. In this and all subsequent figures in this section the nuke is depicted gray.}
\label{fig:P3A}
\end{center}
\end{figure}

\begin{step}{n:P3-nuked}
For every nuked component $C$ of $G-\pmc$ with $\meas(C) \geq 0.1$ and for every $v \in N(C)$ there exists a $P_3$ in $G$ with one endpoint in $v$ and the remaining two vertices in $C$.
\end{step}
\begin{stepproof}
If such a $P_3$ does not exist, by the connectivity of $C$ we have $C \subseteq N(v)$. However, then $\meas(N(v)) \geq \meas(C) \geq 0.1$.
\end{stepproof}

\begin{step}{n:P4}
For every nuked component $C$ of $G- \pmc$, if there exists a vertex $x \in (C \cap X) \setminus N(\pmc \setminus X)$,
then there exists a nonempty set $Z \subseteq N(C) \setminus X$ such that for every $v \in Z$ there exists a $P_4$ in $G$
with one endpoint in $v$ and the remaining three vertices in $C \setminus N(N(C) \setminus (X \cup Z))$.
\end{step}
\begin{stepproof}
See Fig.~\ref{fig:P3B} for an illustration of the proof.
Define $Z$ to be the set of these vertices of $N(C) \setminus X$ that are reachable from $x$ via a path with all internal vertices
in $C \cap D$. The fact that $Z$ is nonempty follows from the facts that $D$ is connected, $\emptyset \neq N(C) \setminus X \subseteq D$, and $x \in X = N(D)$.

Consider any $v \in Z$. Let $P$ be a shortest path from $v$ to $x$ with all internal vertices in $C \cap D$. By the definition of $Z$, such a path exists. Since $P$ is a shortest path, it is an induced one. Furthermore, since $Z \subseteq \pmc$ while $x \notin N(\pmc \setminus X)$, the $P$ contains at least three vertices. Prolong $P$ with a neighbour of $x$ in $D'$, a shelter different than $D$ adjacent to $x$ (whose existence is asserted by~\stepref{n:X-nei}), obtaining a path on at least four vertices with one endpoint in $v$ and remaining vertices in $C$.

To finish the proof, it suffices to argue
that no vertex of $P$ except for $v$ may have a neighbour in $N(C) \setminus (X \cup Z)$.
This statement is true for the part of $P$ contained in $C \cap D$, by the definition of $Z$.
By assumptions, $x$ has no neighbour in $\pmc \setminus X$.
Finally, no vertex in $D'$ is adjacent to any vertex of $N(C) \setminus (X \cup Z) \subseteq D$.
\end{stepproof}

\begin{figure}[tb]
\begin{center}
\includegraphics{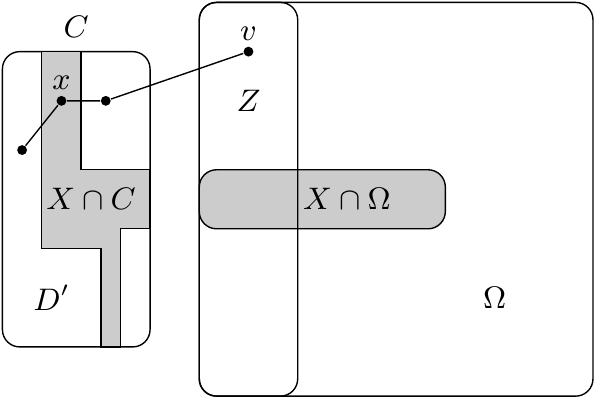}
\caption{Illustration of the proof of Claim~\ref{n:P4}.}
\label{fig:P3B}
\end{center}
\end{figure}

\begin{step}{n:P3-avoid}
For every nuked component $C$ of $G-\pmc$, for every two vertices $u,v \in N(C) \setminus X$, if there exists a vertex $x \in C \cap X \cap (N(v) \setminus N(u))$, then there exists
a $P_3$ in $G$ with one endpoint in $v$ and the remaining two vertices in $C \setminus N(u)$.
\end{step}
\begin{stepproof}
See Fig.~\ref{fig:P3C} for an illustration of the proof.
Let $D'$ be a shelter of $G-X$, different from $D$ and adjacent to $x$, whose existence is asserted by~\stepref{n:nuked-nei}. For the required $P_3$, take the vertices $v$, $x$, and any vertex of $N(x) \cap D'$.
\end{stepproof}

\begin{figure}[tb]
\begin{center}
\includegraphics{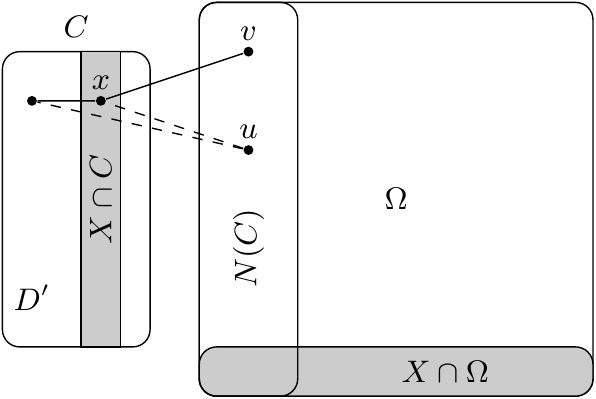}
\caption{Illustration of the proof of Claim~\ref{n:P3-avoid}.}
\label{fig:P3C}
\end{center}
\end{figure}

We now study the possible relations between the neighbourhoods of nuked components. The following observation serves as a starting point.

\begin{step}{n:diff-nei}
For every two nuked components $C_1,C_2$ of $G \setminus \pmc$ it holds that $N(C_1) \setminus X \subseteq N(C_2) \setminus X$ or $N(C_2) \setminus X \subseteq N(C_1) \setminus X$.
\end{step}
\begin{stepproof}
See Fig.~\ref{fig:symdiff} for an illustration of the proof.
By contradiction, assume that there exists $v_i \in N(C_i) \setminus (X \cup N(C_{3-i}))$ for $i=1,2$.
For $i=1,2$, let $P^i$ be a $P_3$ with endpoint in $v_i$ and other vertices in $C_i$, whose existence is asserted by~\stepref{n:P3}.
If $v_1v_2 \in E(G)$, then concatenated paths $P^1$ and $P^2$ form a $P_6$, a contradiction. Otherwise, by Theorem~\ref{thm:pmc}
there exists a component $C$ of $G \setminus \pmc$ with $v_1,v_2 \in N(C)$. Clearly, $C \neq C_i$ for $i=1,2$.
Hence, by concatenating $P^1$, a shortest path from $v_1$ to $v_2$ through $C$, and $P^2$, we obtain an induced path on at least $7$ vertices, a contradiction.
\end{stepproof}

\begin{figure}[tb]
\begin{center}
\includegraphics{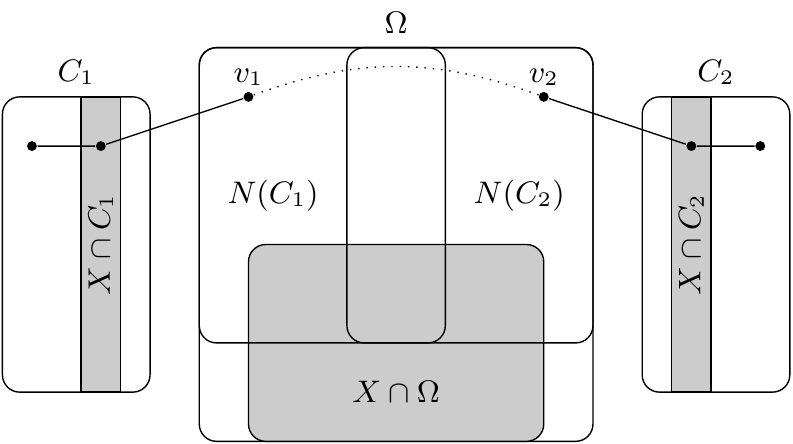}
\caption{Illustration of the proof of Claim~\ref{n:diff-nei}. The dotted connection between $v_1$ and $v_2$
may be realized through a third component.}
\label{fig:symdiff}
\end{center}
\end{figure}

\stepref{n:diff-nei} allows us to order the nuked components of $G-\pmc$ as $C_1,C_2,\ldots,C_r$, such that
$$N(C_1) \setminus X \supseteq N(C_2) \setminus X \supseteq \ldots \supseteq N(C_r) \setminus X.$$
By~\stepref{n:two-nuked}, $r \geq 2$.

We say that two nuked components $C_i$ and $C_j$, $1 \leq i,j \leq r$, $i \neq j$ are \emph{linked}
if for every choice of $u \in N(C_i) \setminus X$ and $v \in N(C_j) \setminus X$ there exists an induced path in $G$
with endpoints $u$ and $v$ and all internal vertices in $V(G) \setminus N[C_i \cup C_j]$.
We remark that if $u=v$ or $uv \in E(G)$, then the last assertion is true, as we can take an one- or two-vertex path, respectively.

In the next few observations we investigate the properties of linked components.
\begin{step}{n:linked1}
If $C_i$ and $C_j$ are linked, then for every two vertices $u,v \in \pmc \setminus X$, one of the following holds:
\begin{enumerate}
\item $N(u) \cap X \cap C_i = N(v) \cap X \cap C_i$,
\item $N(u) \cap X \cap C_j = N(v) \cap X \cap C_j$,
\item $N(u) \cap X \cap (C_i \cup C_j) \subsetneq N(v) \cap X \cap (C_i \cup C_j)$, or
\item $N(v) \cap X \cap (C_i \cup C_j) \subsetneq N(u) \cap X \cap (C_i \cup C_j)$.
\end{enumerate}
\end{step}
\begin{stepproof}
See Fig.~\ref{fig:linked1} for an illustration of the proof.
Assume the contrary. By symmetry, we can consider only the case where $(N(v) \cap C_i \cap X) \setminus N(u) \neq \emptyset$
and $(N(u) \cap C_j \cap X) \setminus N(v) \neq \emptyset$. Clearly, $v \in N(C_i) \setminus X$, $u \in N(C_j) \setminus X$, and $u \neq v$.
By applying~\stepref{n:P3-avoid} twice, we obtain a $P_3$ $P^v$ with
endpoint in $v$ and the remaining two vertices in $C_i \setminus N(u)$, and a $P_3$ $P^u$ with endpoint in $u$ and the remaining two
vertices in $C_j \setminus N(v)$. These two paths, together with the induced path between $u$ and $v$ promised by the fact that $C_i$ and $C_j$
are linked, yield an induced path on at least six vertices, a contradiction.
\end{stepproof}

\begin{figure}[tb]
\begin{center}
\includegraphics{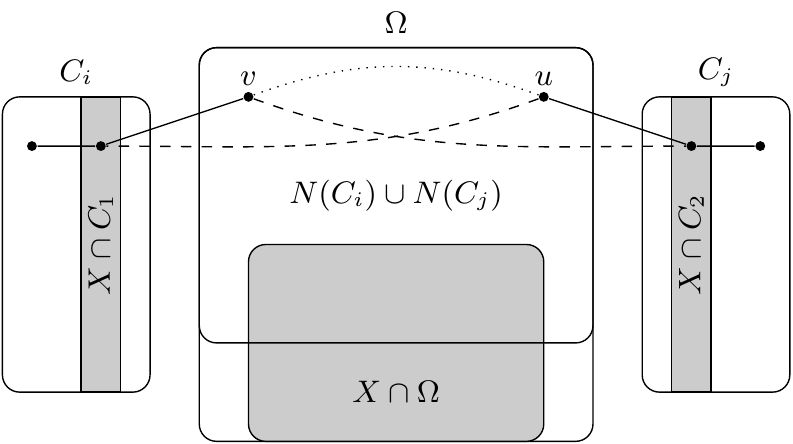}
\caption{Illustration of the proof of Claim~\ref{n:linked1}. The existence dotted connection between $v$ and $u$
  is implied by the linkedness between $C_i$ and $C_j$.}
\label{fig:linked1}
\end{center}
\end{figure}

With every nuked component $C_i$ we associate the family $\FF_i := \{N(v) \cap C_i \cap X: v \in \pmc \setminus X\}$.
\begin{step}{n:linked2}
If $C_i$ and $C_j$ are linked, then either $\FF_i$ or $\FF_j$ has unique maximal element with respect to inclusion.
\end{step}
\begin{stepproof}
Assume otherwise. Let $u,v \in \pmc \setminus X$ be such that $A_u := N(u) \cap C_i \cap X$ and $A_v := N(v) \cap C_i \cap X$ are two different
maximal elements of $\FF_i$,
and $p,q \in \pmc \setminus X$ be such that $B_p := N(p) \cap C_j \cap X$ and $B_q := N(q) \cap C_j \cap X$ are two different maximal elements of $\FF_j$.
By~\stepref{n:linked1} we have $N(u) \cap C_j \cap X = N(v) \cap C_j \cap X$; let us denote this set $B$.
Similarly, $N(p) \cap C_i \cap X = N(q) \cap C_j \cap X$, and we denote this set $A$.
If $B$ and $B_p$ are incomparable with respect to inclusion, then~\stepref{n:linked1} asserts that $A = A_u$ (for pair $u$ and $p$) and $A = A_v$ (for pair $v$ and $p$), a contradiction. By maximality of $B_p$, we have $B \subseteq B_p$.
Similarly we infer that $B \subseteq B_q$. Hence, $B \subseteq B_p \cap B_q$; by the incomparability of $B_p$ and $B_q$, we infer that
$B \subsetneq B_p$. However,~\stepref{n:linked1} asserts then that
$A_u \subseteq A$ (for the pair $p,u$) and $A_v \subseteq A$ (for the pair $p,v$).
This is a contradiction with the maximality and incomparability of $A_u$ and $A_v$.
\end{stepproof}
\begin{step}{n:linked3}
Let $I \subseteq \{1,2,\ldots,r\}$ be the set of indices such that for any $i,j \in I$, $i \neq j$, $C_i$ and $C_j$ are linked.
Then there exists a vertex $v \in \pmc \setminus X$ and an index $i_0$ such that
$$X \cap N(\pmc \setminus X) \cap \bigcup_{i \in I \setminus \{i_0\}} C_i \subseteq N(v).$$
\end{step}
\begin{stepproof}
If $|I| \leq 1$, the claim is straightforward, so assume otherwise.
By~\stepref{n:linked2}, there exists at most one index $i_0$ such that $\FF_{i_0}$ does not admit a unique maximal element.
(If no such index exists, we set $i_0 \in I$ arbitrarily). 

For $u \in \pmc \setminus X$, we define $I_u \subseteq I \setminus \{i_0\}$ to be the set of these indices $i$ for which $N(u) \cap C_i \cap X$
is the maximal element of $\FF_i$.
Let $v$ be such a vertex that $|I_v|$ is maximized. To finish the proof it suffices to show that $I_v = I \setminus \{i_0\}$.
Assume the contrary: there exists $j \in I \setminus \{i_0\}$ such that $N(v) \cap C_j \cap X$ is not the maximal element
of $\FF_j$. Let $w \in \pmc \setminus X$ be such that $N(w) \cap C_j \cap X$ is the maximal element of $\FF_j$.
We have $N(v) \cap C_j \cap X \subsetneq N(w) \cap C_j \cap X$.
By~\stepref{n:linked1}, for every $i \in I_v$ we have $N(v) \cap C_i \cap X \subseteq N(w) \cap C_i \cap X$. However, $N(v) \cap C_i \cap X$
is the unique maximal element of $\FF_i$. Consequently, $I_v \subseteq I_w$. However, $j \in I_w \setminus I_v$, which contradicts the
choice of $v$.
\end{stepproof}

Consider now the following corollary of~\stepref{n:P4}.
\begin{step}{n:no-P4}
For every $2 \leq i \leq r$ we have $X \cap C_i \subseteq N(\pmc \setminus X)$. Furthermore,
if $X \cap C_1 \not\subseteq N(\pmc \setminus X)$, then the set $Z$ whose existence is asserted in~\stepref{n:P4} for the component $C_1$
is completely contained in $N(C_1) \setminus (X \cup N(C_2))$.
\end{step}
\begin{stepproof}
See Fig.~\ref{fig:no-P4} for an illustration of the proof.
Assume the contrary. By~\stepref{n:P4} there exists an index $1 \leq i \leq r$ and a $P_4$ in $G$ with one endpoint $v \in N(C_2) \setminus X$
and the remaining three vertices in $C_i$; denote this path $P$. Let $j \in \{1,2\} \setminus \{i\}$. By~\stepref{n:P3}, there exists
a $P_3$ $Q$ with endpoint $v$ and the remaining vertices in $C_j$. However, the concatenation of $P$ and $Q$ is a $P_6$ in $G$, a contradiction.
\end{stepproof}

\begin{figure}[tb]
\begin{center}
\includegraphics{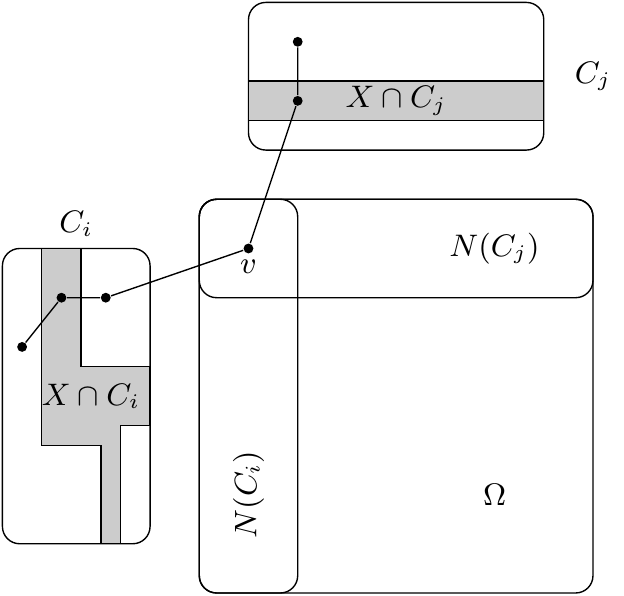}
\caption{Illustration of the proof of Claim~\ref{n:no-P4}. The set $X \cap \Omega$ is omitted in order to keep the picture readable.}
\label{fig:no-P4}
\end{center}
\end{figure}

Observe now the following.
\begin{step}{n:linked4}
For every $2 \leq i,j \leq r$, $i \neq j$, the components $C_i$ and $C_j$ are linked.
\end{step}
\begin{stepproof}
Two vertices $u \in N(C_i) \setminus X$ and $v \in N(C_j) \setminus X$ can be linked either via a direct edge if it exists in $G$,
or via a shortest path with internal vertices in $C_1$.
\end{stepproof}
Combining now~\stepref{n:no-P4} with~\stepref{n:linked3} applied to $I = \{2,3,\ldots,r\}$ we obtain that
\begin{step}{n:two-left}
There exists an index $2 \leq i_0 \leq r$ such that $\meas(C_1 \cup C_{i_0}) \geq 0.8$.
In particular, $\meas(C_1),\meas(C_{i_0}) \geq 0.3$.
\end{step}
\begin{stepproof}
By~\stepref{n:linked3}, applied to $I=\{2,3,\ldots,r\}$, we have an index $i_0$ and a vertex $v$ adjacent to all vertices
of $X \cap C_j \cap N(\pmc \setminus X)$ for $j \notin \{1,i_0\}$. However, by~\stepref{n:no-P4}, these are actually all vertices of
$X \cap C_j$. Since $\meas(N(v)) \leq \nukhit \leq 0.1$ and $\meas(\pmc) \leq 0.1$, the first claim follows.
The second claim follows from the choice of $\pmc$: $\meas(C) \leq 0.5$ for every connected component $C$ of $G-\pmc$.
\end{stepproof}
Fix the index $i_0$ from~\stepref{n:two-left}.
\begin{step}{n:free-w}
$N(C_1) \cup N(C_{i_0}) \neq \pmc$.
\end{step}
\begin{stepproof}
See Fig.~\ref{fig:free-w} for an illustration of the proof.
By contradiction, assume that $N(C_1) \cup N(C_{i_0}) = \pmc$.
By Theorem~\ref{thm:pmc}, neither $N(C_1)$ nor $N(C_{i_0})$ equals the whole $\pmc$, thus 
there exists $v \in N(C_1) \setminus N(C_{i_0})$ and $u \in N(C_{i_0}) \setminus N(C_1)$.
By~\stepref{n:P3-nuked}, there exist a $P_3$ $P^v$ with endpoint in $v$ and remaining two vertices in $C_1$,
and a $P_3$ $P^u$ with endpoint in $u$ and remaining two vertices in $C_{i_0}$.
If $uv \in E(G)$, then these two paths together give a $P_6$ in $G$, a contradiction.
Otherwise, by Theorem~\ref{thm:pmc}, there exists a component $C$ of $G-\pmc$ such that $u,v \in N(C)$.
Clearly, $C \notin \{C_1,C_{i_0}\}$. However, then a concatenation of $P^v$, a shortest path from $v$ to $u$
with internal vertices in $C$, and $P^u$, yields an induced path in $G$ on at least $7$ vertices, a contradiction.
\end{stepproof}

\begin{figure}[tb]
\begin{center}
\includegraphics{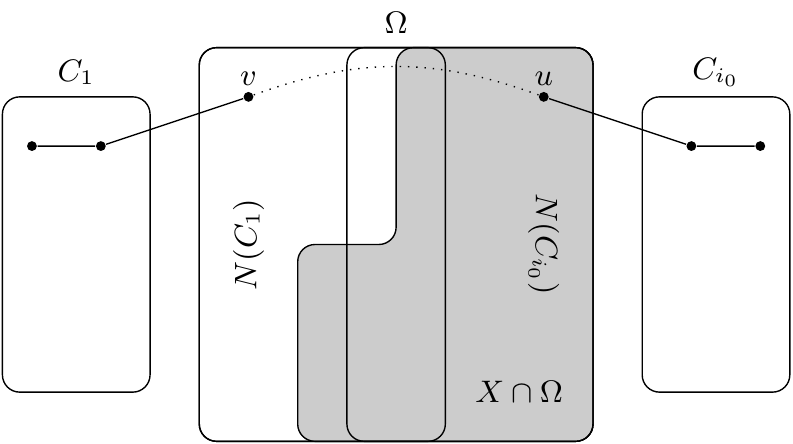}
\caption{Illustration of the proof of Claim~\ref{n:free-w}. The dotted connection between $v$ and $u$
may be realized through a third component.}
\label{fig:free-w}
\end{center}
\end{figure}
Let $w$ be an arbitrary vertex of $\pmc \setminus N(C_1 \cup C_{i_0})$.
\begin{step}{n:two-linked}
$C_1$ and $C_{i_0}$ are linked.
\end{step}
\begin{stepproof}
Consider any $v \in N(C_1) \setminus X$ and $u \in N(C_{i_0}) \setminus X$; we are going to exhibit an induced path from $v$ to $u$
with internal vertices in $V(G) \setminus N[C_1 \cup C_{i_0}]$.
If $v = u$ or $vu \in E(G)$, then we are done with the one- or two-vertex path.
If there exists a connected component $C$ of $G-\pmc$ different than $C_1$ or $C_{i_0}$ such that $u,v \in N(C)$, then
we can choose a shortest path from $v$ to $u$ with all internal vertices in $C$.

Otherwise, we route the path through the vertex $w$.
Let $P^v = vw$ if $vw \in E(G)$, and otherwise let $P^v$ be a shortest path from $v$ to $w$ with internal vertices in a
connected component $C$ covering the nonedge $vw$; note that $C \notin \{C_1,C_{i_0}\}$ as $w \notin N(C_1 \cup C_{i_0})$.
Similarly define the path $P^u$ from $u$ to $v$. Since no component different than $C_1$ or $C_{i_0}$ has both $u$ and $v$
in their neighbourhood, the concatenation of $P^v$ and $P^u$ forms the desired path.
\end{stepproof}

In the next two claims we exhibit the final contradiction.
\begin{step}{n:no-Z}
$C_1 \cap X \subseteq N(\pmc \setminus X)$.
\end{step}
\begin{stepproof}
See Fig.~\ref{fig:no-Z} for an illustration of the proof.
Assume the contrary. By~\stepref{n:no-P4}, the set $Z$ whose existence is asserted by~\stepref{n:P4} for the component $C_1$
is completely contained in $N(C_1) \setminus (X \cup N(C_2)) \subseteq N(C_1) \setminus (X \cup N(C_{i_0}))$.
Consider any $v \in Z$ and $u \in N(C_{i_0}) \setminus X \subseteq N(C_1) \setminus (X \cup Z)$.
By~\stepref{n:P4}, there exists a $P_4$ $P^v$ with endpoint in $v$ and internal vertices in $C_1 \setminus N(u)$. Furthermore, by~\stepref{n:P3-nuked},
there exists a $P_3$ $P^u$ with endpoint $u$ and internal vertices in $C_{i_0}$. Recall that $v \notin N(C_{i_0})$,
thus $P^u$ does not contain any neighbour of $v$, except for possibly $u$.
Hence, the paths $P^v$ and $P^u$, together with the path between $v$ and $u$ whose existence is asserted by the fact that $C_1$
and $C_{i_0}$ are linked, form an induced path in $G$ on at least seven vertices, a contradiction.
\end{stepproof}

\begin{figure}[tb]
\begin{center}
\includegraphics{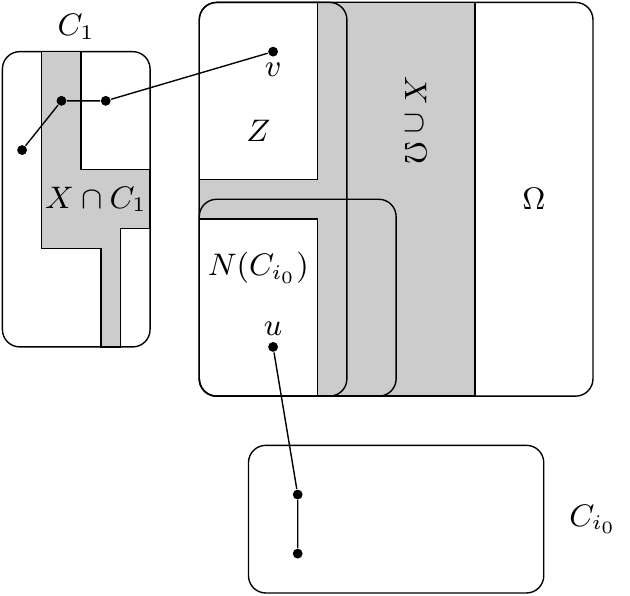}
\caption{Illustration of the proof of Claim~\ref{n:no-Z}.
The vertex $u$ is not adjacent to any of the three vertices in $C_1$ since $u \notin Z$ (\stepref{n:no-P4}).
The existence of a connection between $u$ and $v$ is guaranteed by the linkedness of $C_1$ and $C_{i_0}$.
}  
\label{fig:no-Z}
\end{center}
\end{figure}

\begin{step}{n:see-half}
There exists a vertex $v$ such that $C_1 \cap X \subseteq N(v)$ or $C_{i_0} \cap X \subseteq N(v)$.
\end{step}
\begin{stepproof}
By~\stepref{n:linked3} applied to $I = \{1,i_0\}$, we obtain a vertex $v$ such that
$C_1 \cap X \cap N(\pmc \setminus X) \subseteq N(v)$ or $C_{i_0} \cap X \cap N(\pmc \setminus X) \subseteq N(v)$.
However, $C_{i_0} \cap X \subseteq N(\pmc \setminus X)$ due to~\stepref{n:no-P4} and $C_1 \cap X \subseteq N(\pmc \setminus X)$ due to~\stepref{n:no-Z}.
\end{stepproof}

The last claim is in contradiction with~\stepref{n:two-left}, asserting that $\meas(C_1),\meas(C_{i_0}) \geq 0.3$.
This finishes the proof of Theorem~\ref{thm:hit-nuke}.

\section{The Algorithm for \ISname}\label{sec:ISalg}

We now make use of Theorems~\ref{thm:hit-pmc} and~\ref{thm:hit-nuke} to design an algorithm
that solves \ISname{} in $n$-vertex $P_6$-free graphs in $n^{\Oh(\log^2 n)}$ time.

\subsection{Description of the algorithm}

The algorithm consists of two recursive procedures, \texttt{FindIS} and \texttt{FindISNuke}, which both aim to find an independent set of maximum weight in a given connected vertex-weighted $P_6$-free graph~$G$. 
The procedure \texttt{FindIS} is the `base' procedure, which we call on the graph $G$. Both procedures make recursive calls to themselves and to each other. We describe each procedure, and then analyze their running time.

\subsubsection{Procedure \texttt{FindIS}}

The input for the procedure \texttt{FindIS} is just a connected $P_6$-free graph $G$.
As a base case, if the input graph consists of one vertex, \texttt{FindIS} returns the weight of this vertex.
Otherwise, it checks if there exists a vertex of degree at least $0.05\pmchit |V(G)|$, where the constant $\pmchit$ comes
from Theorem~\ref{thm:hit-pmc}.

If such a vertex $v$ exists, then the procedure branches on the vertex $v$.
In one branch, we seek a solution not containing $v$, and we call \texttt{FindIS} independently on every connected component
of $G-v$. In the second branch, we seek a solution containing $v$, and we call \texttt{FindIS} independently on every
connected component of $G-N(v)$.

Otherwise, that is if all vertices are of degree less than $0.05\pmchit |V(G)|$, 
the algorithm takes an arbitrary minimal triangulation $\cG$ of $G$ (see e.g.~\cite{pinar-survey} for algorithms that find such a triangulation), constructs
its clique tree, and finds a maximal clique $\pmc$ in $\cG$ such that every connected component of $G-\pmc$ has at most $|V(G)|/2$
vertices (such a maximal clique exists by standard arguments). 
We observe the following:

\begin{step}{a:small-pmc}
$|\pmc| < 0.05|V(G)|$
\end{step}
\begin{stepproof}
If $|\pmc| \geq 0.05|V(G)|$, then Theorem~\ref{thm:hit-pmc} applied to $\pmc$ with the uniform measure, implies that there exists
a vertex $v$ with $|N(v)| \geq |N(v) \cap \pmc| \geq 0.05\pmchit |V(G)|$, a contradiction.
\end{stepproof}
By the choice of $\pmc$, if we set $\nukth = 0.8 |V(G)|$, then $\pmc$ is a $(0.1,\nukth)$-nuke in $G$ (with a lot of slack
in the inequalities in the second and third point in the definition of a nuke).
The algorithm passes the graph $G$, the threshold $\nukth$, and the nuke $\pmc$ to the procedure \texttt{FindISNuke}.

\subsubsection{Procedure \texttt{FindISNuke}}

The input for the procedure \texttt{FindISNuke} is a connected $P_6$-free graph $G$, a threshold $\nukth$,
and a set $X \subseteq V(G)$ with the promise that for every connected component $C$ of $G-X$ it holds that $|C|+|X| \leq \nukth$.

The algorithm first checks if $G$ contains at least two vertices and $X$ is a $(0.1,\nukth)$-nuke of $G$ (note that $X$ is such a nuke when \texttt{FindISNuke} is invoked by \texttt{FindIS}).
If this is not the case, then the algorithm invokes \texttt{FindIS} on the graph $G$, forgetting about $\nukth$ and $X$.
Otherwise, it finds any inclusion-wise minimal $(0.1,\nukth)$-nuke $Y \subseteq X$, and finds a vertex $v$ with
$|N(v) \cap Y| \geq \nukhit |Y|$; the existence of such vertex is guaranteed by applying Theorem~\ref{thm:hit-nuke} to $Y$ with the uniform measure.
The algorithm branches on vertex $v$ as usual.
That is, in one branch, we seek a solution not containing $v$,
and we call \texttt{FindISNuke} independently on every connected component
of $G-v$. In the second branch, we seek a solution containing $v$, and we call \texttt{FindISNuke} independently on every
connected component of $G-N(v)$.
In every subcall, we pass the same threshold $\nukth$, and the set $Y$ restricted to the vertex set of the connected component in question.
Clearly, since we delete only vertices from $G$ or reduce $X$ to a minimal sub-nuke, in the subcalls we maintain
the promise that for every connected component $C$ of $G-X$ it holds that $|C|+|X| \leq \nukth$.

\subsection{Analysis}

As the algorithm performs exhaustive branching, it clearly returns an optimum solution.
Also, the polynomial space bound is immediate.
It remains to argue about the running time.

Consider the recursion tree $\mathcal{T}_0$ of the algorithm, and focus on one call $\mathfrak{c}$ to $\mathtt{FindIS}(G)$
that resulted in a subcall $\mathtt{FindISNuke}(G,\nukth,X)$; here $\nukth = 0.8|V(G)|$ and $X$ is a potential maximal clique in $G$
of size at most $0.05|V(G)|$ (by~\stepref{a:small-pmc}). 
Every call to \texttt{FindISNuke} results either in branching and multiple calls to the same procedure (call it a \emph{branching call}), or
a single call to \texttt{FindIS} (call it a \emph{fallback call}). 
Let $\mathcal{T}$ be a maximal subtree at $\mathcal{T}_0$, rooted at the chosen call $\mathfrak{c}$ to $\mathtt{FindIS}$, that contains
(apart from the root) only calls to \texttt{FindISNuke}. That is, we put into $\mathcal{T}$ all recursive calls that originated
from $\mathfrak{c}$, and stop whenever we encounter a fallback call; in particular, all leaves of $\mathcal{T}$ are fallback calls.

First, observe that $\mathcal{T}$ has $|V(G)|^{\Oh(\log |V(G)|)}$ leaves by standard analysis: in every branch either we delete
one vertex from $G$, or delete a constant fraction of the minimal sub-nuke of $X$, while independently considering every connected component
only helps in the process.

Let $\mathtt{FindISNuke}(G',\nukth',X')$ be a leaf of $\mathcal{T}$. We claim the following.
\begin{step}{a:nuke-decrease}
$|V(G')| < \frac{8}{9} |V(G)|$.
\end{step}
\begin{stepproof}
Since we are considering a fallback call, either $|V(G')|=1$
or $X'$ is not a $(0.1,\nukth')$-nuke of $G'$.
In the first case, since $|V(G)| > 1$, the claim is obvious.
In the second case, consider the reasons why
$X'$ may not be a $(0.1,\nukth')$-nuke of $G'$. Clearly, $\nukth' = \nukth = 0.8|V(G)|$ and, by the promise maintained
in the course of algorithm, for every connected component $C$ of $G'-X'$ it holds that $|C| + |X'| \leq \nukth$.
Furthermore, $(1-2\cdot 0.1)|V(G')| \leq (1-2\cdot 0.1)|V(G)| = \nukth$.

Hence, either $(1-0.1)|V(G')| < \nukth = 0.8|V(G)|$ or $|X'| > 0.1|V(G')|$.
In the first case $|V(G')| < \frac{8}{9}|V(G)|$, while in the second case $|V(G')| \leq |V(G)|/2$,
because $|X| \leq 0.05|V(G)|$ and $X' \subseteq X$.
\end{stepproof}
By~\stepref{a:nuke-decrease}, if we contract every such subtree $\mathcal{T}$ to a single super-node of the recursion tree $\mathcal{T}_0$,
then at each such super-node we branch into $|V(G)|^{\Oh(\log |V(G)|)}$ subcases, and in each subcase decrease the number of vertices by a multiplicative factor.

Now focus on a call $\mathfrak{c}$ to $\mathtt{FindIS}$ that branches on a vertex $v \in V(G)$ of degree
at least $0.05\pmchit |V(G)|$.
Observe that at most one recursive subcall of $\mathfrak{c}$ is invoked on a graph with at least $(1-0.05\pmchit)|V(G)|$
vertices: the one for the largest connected component of $G-v$.
Mark the edges of the recursion tree that correspond to such subcalls.
The marked edges form vertex-disjoint top-bottom paths in the recursion tree. 
If we contract them (along with the aforementioned subtrees $\mathcal{T}$), we obtain a recursion tree where every node has
$n^{\Oh(\log n)}$ subcases and where in each subcase the number of vertices decreases by a constant factor. 
Consequently, the size of the recursion tree is $n^{\Oh(\log^2 n)}$.
This finishes the analysis of the algorithm, and concludes the proof of Theorem~\ref{thm:main}.

\section{The Algorithm for \EDname}\label{sec:EDalg}

In this section we prove Theorem~\ref{thm:effdom}.
The overall approach is as follows: we take any minimal triangulation of the input graph $G$, and perform
the standard dynamic programming algorithm on the clique tree of this completion (which is a tree decomposition of $G$).
In this standard dynamic programming algorithm, every state at bag $B$ keeps information about which vertices of $B$ are
contained in the constructed efficient dominating set, and which vertices of $B$ has been already dominated by the forgotten
parts of the graph.

The main insight is that we can use Theorem~\ref{thm:hit-pmc}, together with technical insight from Section~\ref{sec:nukes},
to show that in $P_6$-free graphs there are only polynomially many reasonable states for the aforementioned dynamic programming algorithm,
yielding the claimed polynomial running time.

\subsection{Bounding the Number of States}

Before we state this main result formally, we need the following definition. Let
$\pmc$ be a potential maximal clique in $G$, and let $\mathcal{C}$ be the set of connected components of $G-\pmc$.
A \emph{state} is a function $f:\pmc \to \mathcal{C} \cup \{\pmc, \bot\}$. A state $f$ is \emph{consistent} with an efficient
dominating set $X$ if $X \cap \pmc = f^{-1}(\bot)$ and furthermore, for every $v \in \pmc \setminus X$, the unique vertex
of $N(v) \cap X$ belongs to the vertex set of $f(v)$. 

\begin{theorem}\label{thm:effdom-states}
Given a $P_6$-free graph $G$ and a potential maximal clique $\pmc$ in $G$, one can in polynomial time 
compute a family $\mathcal{S}$ of states of polynomial size, such that for every efficient dominating set $X$ in $G$,
there exists a state $f \in \mathcal{S}$ consistent with $X$.
\end{theorem}

This section is devoted to the proof of Theorem~\ref{thm:effdom-states}.
We describe the algorithm as a branching algorithm that outputs a state at every leaf of the branching tree, 
and every leaf-to-root path of the branching tree contains $\Oh(\log n)$ nodes of constant degree and $\Oh(1)$ nodes
of degree polynomial in $n$. Furthermore, it will be straightforward to perform the computation required 
at every node of the branching tree in polynomial time. These properties give the promised polynomial bounds on the size of the output
and the total running time.

Every node of the branching tree is labeled with two vertex sets $X_0$ and $Y$, and the goal of the subtree rooted at the
node labeled $(X_0,Y)$ is to output a family of states such that for every efficient dominating set $X$ with $X_0 \subseteq X$
and $(X \setminus X_0) \subseteq Y$ (henceforth called an efficient dominating set \emph{consistent} with $(X_0,Y)$) there exists an output consistent state.
In every branching step, in every subcase, the algorithm puts some vertices into $X_0$ and/or removes some vertices from $Y$.
Since every two elements of an efficient dominating set are within distance at least three, we implicitly assume that if the algorithm
puts a vertex $v$ into $X_0$, it at the same time removes from $Y$ all vertices within distance at most two from $v$.
Furthermore, we immediately terminate a branch if two vertices of $X_0$ are within distance less than three,
or if there exists $v \in V(G)$ with $N[v] \cap (X_0 \cup Y) = \emptyset$.

The algorithm terminates branching at nodes labeled $(X_0,Y)$ where for every $v \in \pmc$ either $N[v]$ contains a vertex of $X_0$, or $N[v] \cap Y$ is contained
in a single component of $\mathcal{C}$.
For such a label $(X_0,Y)$, we define a state $f$ as follows:
$f(v) = \bot$ for $v \in X_0 \cap \pmc$, $f(v) = \pmc$ for $v \in N(X_0) \cap \pmc$,
and otherwise $f(v)$ is the unique component of $\mathcal{C}$ that contains vertices of $N[v] \cap (X_0 \cup Y)$.
It is straightforward to verify that if $X$ is consistent with $(X_0,Y)$, then $f$ is well-defined and it is also consistent with $f$. 
Consequently, the algorithm outputs the function $f$ in this leaf node of the branching tree.

At the root of the branching tree we have $X_0 = \emptyset$ and $Y = V(G)$.

\subsubsection{Guessing Vertices from the Solution Inside the PMC}

We start with the following observation.
\begin{lemma}\label{lem:ed:in-pmc}
For every $P_7$-free graph $G$, every potential maximal clique $\pmc$ in $G$, and every efficient dominating set $X$ in $G$,
    we have $|\pmc \cap X| \leq 1/\pmchit$, where the constant $\pmchit$ comes from Theorem~\ref{thm:hit-pmc}.
\end{lemma}
\begin{proof}
Without loss of generality, we can assume that $G$ is connected (we can consider every component independently)
  and contains at least two vertices (for one-vertex graphs the statement is trivial).

Let $\ell = |\pmc \cap X|$. Consider a measure $\meas$ on $\pmc$ such that $\meas(v) = 1/\ell$ for every $v \in \pmc \cap X$
and $\meas(v) = 0$ otherwise. By Theorem~\ref{thm:hit-pmc}, there exists a vertex $u$ with $\meas(N(u)) \geq \pmchit$.
However, by the definition of an efficient dominating set, we have $|N(u) \cap X| \leq 1$. Consequently, 
$\meas(N(u)) \leq 1/\ell$, hence $\ell \leq 1/\pmchit$.
\end{proof}

By Lemma~\ref{lem:ed:in-pmc}, our algorithm can, as a first step, guess all vertices from the solution that lie in $\pmc$.
More formally, the algorithm branches into a subcase for every subset $X_\pmc \subseteq \pmc$ of size at most $1/\pmchit$;
we label the subcase corresponding to $X_\pmc$ by $(X_\pmc, V(G) \setminus (N^2[X_\pmc] \cup \pmc))$. 
We emphasize here that we not only removed from $Y$ all vertices within distance at most two from $X_\pmc$, but also all vertices from $\pmc$.
Thus, from this point, we have that $Y \cap \pmc = \emptyset$.

\subsubsection{Reduction Rule}

Fix a node of the branching tree labeled $(X_0,Y)$ with $Y \cap \pmc = \emptyset$.
We say that a component $C \in \mathcal{C}$ is \emph{active} if $C \cap Y \neq \emptyset$. 
Let $A = \pmc \setminus N[X_0]$ be the set of vertices that are not yet dominated by the vertices from $X_0$.
Let $B \subseteq A$ be the set of these vertices $v$ such that the vertices of $N(v) \cap Y$ appear in at least two connected components
of $\mathcal{C}$. Note that the algorithm terminates branching and outputs a state if $B = \emptyset$; the main goal in the branching step
is to shrink the set $B$ as much as possible.

We start by introducing a reduction rule, aimed at shrinking the set $Y$ without performing any branching.
For a vertex $u \in V(G) \setminus \pmc$, let $C(u)$ be the component of $\mathcal{C}$ that contains $u$.
Assume that for some vertex $v \in B$ there exists $u \in N(v) \cap Y$ such that $N[u] \cap Y \subseteq N(v)$. 
Let $X$ be an efficient dominating set consistent with $(X_0,Y)$. Since $Y \cap N[X_0] = \emptyset$, the vertex $u$
needs to be dominated by some vertex $w \in N[u] \cap (X \setminus X_0) \subseteq N[u] \cap Y$. By our assumption,
$w$ also dominates $v$. Consequently, in every efficient dominating set consistent with $(X_0,Y)$, the vertex $v$
is dominated by an element $C(u)$, and we can introduce the following reduction rule.

\medskip
\noindent\textbf{Reduction Rule.} If there exist vertices $v \in B$ and $u \in N(v) \cap Y$ such that $N[u] \cap Y \subseteq N(v)$,
  then remove from $Y$ all vertices of $N(v) \setminus C(u)$.

\medskip

Note that, in particular, the aforementioned Reduction Rule triggers if some vertex of $B$ is fully adjacent to an active component
(recall that $\pmc \cap Y = \emptyset$).

In what follows we assume that at every node of the recursion tree, the Reduction Rule is applied exhaustively.
Observe that if this rule is not applicable, then for every $v \in B$ and $u \in N(v) \cap Y$, there exists a vertex $w \in (Y \cap N(u)) \setminus N(v)$;
note that $w \in C(u)$ and $\{v,u,w\}$ induce a $P_3$ in $G$. The main intuition of the remaining proof is that the graph needs to be highly structured in order to not to allow two such $P_3$'s to ``glue'' together into a $P_6$ in $G$.

\subsubsection{Structure of $B$-Neighbourhoods}

As a first application of this principle, observe the following.
\begin{lemma}\label{lem:Cuniv}
If $C^1,C^2$ are two different components of $\mathcal{C}$, then $N(C^1) \setminus N(C^2)$ is fully adjacent to 
$C^1$, or $N(C^2) \setminus N(C^1)$ is fully adjacent to $C^2$.
\end{lemma}
\begin{proof}
See Fig.~\ref{fig:ed:symdiff} for an illustration of the proof.
Assume the contrary. Let $v^i \in N(C^i) \setminus N(C^{3-i})$ be a vertex that is not fully adjacent to $C^i$ for $i=1,2$.
Since $v^i$ is not fully adjacent to $C^i$, but $v^i \in N(C^i)$ and $C^i$ is connected, there exists an induced $P_3$ with 
one endpoint $v^i$ and other vertices in $C^i$; denote this $P_3$ as $P^i$. Furthermore, by Theorem~\ref{thm:pmc},
either $v^1v^2 \in E(G)$ or there exists a component $C \in \mathcal{C}$ such that $v^1,v^2 \in N(C)$. Clearly, $C \notin \{C^1,C^2\}$.
Consequently, by concatenating $P^1$, $P^2$, and the edge $v^1v^2$ or a shortest path between $v^1$ and $v^2$ with internal vertices
in $C$, we obtain an induced path on at least $6$ vertices, a contradiction.
\end{proof}

\begin{figure}[tb]
\begin{center}
\includegraphics{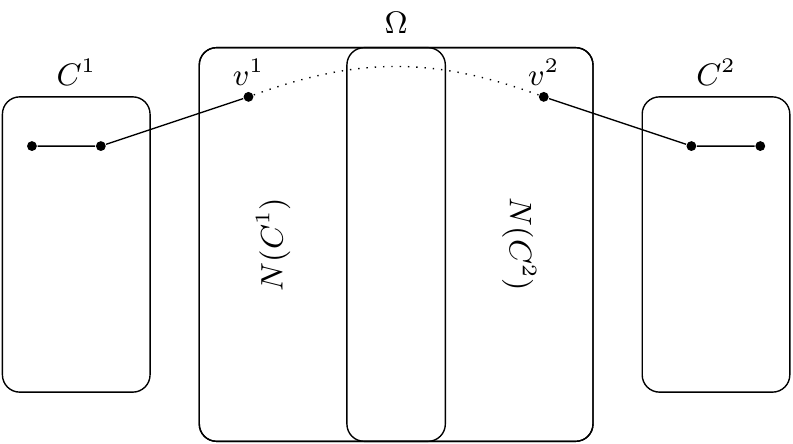}
\caption{Illustration of the proof of Lemma~\ref{lem:Cuniv}. The dotted connection between $v^1$ and $v^2$
may be realized through a third component.}
\label{fig:ed:symdiff}
\end{center}
\end{figure}

Since our Reduction Rule removes from $B$ vertices that are fully adjacent to some active component,
we infer that we can enumerate active components as $C_1,C_2,\ldots,C_r$ such that $N_B(C_i) \supseteq N_B(C_j)$ for every $i \leq j$.
Furthermore, since every element in $B$ has neighbours in $Y$ in at least two components by definition,
we have that $N_B(C_1) = N_B(C_2) = B$. Summing up,
\begin{equation}\label{eq:EDchain}
B = N_B(C_1) = N_B(C_2) \supseteq N_B(C_3) \supseteq N_B(C_4) \supseteq \ldots \supseteq N_B(C_r).
\end{equation}

\subsubsection{Obtaining Linkedness}

In order to ``glue'' two $P_3$'s, we use the following notion. We say that two active components $C^1$ and $C^2$ are \emph{linked} if for every two vertices $v^1 \in B \cap N(C^1)$, $v^2 \in B \cap N(C^2)$ there exists an induced path in $G$ with endpoints $v^1$ and $v^2$
and all internal vertices in $V(G) \setminus (N[C^1] \cup N[C^2])$. We explicitly allow 1-vertex and 2-vertex paths here (if $v^1 = v^2$ or $v^1v^2 \in E(G)$).

We start by observing the following:
\begin{lemma}\label{lem:linked1}
Every pair of active components is linked, except for possibly the pair $\{C_1,C_2\}$.
\end{lemma}
\begin{proof}
By~\eqref{eq:EDchain}, for every other pair $\{C_i,C_j\}$, we can use either $C_1$ or $C_2$ to route the desired path.
\end{proof}

Our goal now is to ensure that also $\{C_1,C_2\}$ are linked.
The following lemma uses essentially the same arguments as Claim~\ref{n:two-linked} of Section~\ref{sec:nukes}.

\begin{lemma}\label{lem:linked2}
If two active components $C^1$ and $C^2$ satisfy $N(C^1) \cup N(C^2) \neq \pmc$, then they are linked.
\end{lemma}
\begin{proof}
Let $w \in \pmc \setminus (N(C^1) \cup N(C^2))$ and consider two vertices $v^1 \in B \cap N(C^1)$, $v^2 \in B \cap N(C^2)$.
If $v^1 = v^2$ or $v^1v^2 \in E(G)$, then we are trivially done.
Furthermore, if there exists a component $C \in \mathcal{C} \setminus \{C^1,C^2\}$ with $v^1,v^2 \in N(C)$, then we are done as well
by taking a shortest path from $v^1$ to $v^2$
with all internal vertices in $C$.

In the remaining case, we start with connecting for $i=1,2$ the vertex $v^i$ with $w$ by an induced path $P^i$ as follows:
if $v^iw \in E(G)$, then we take $P^i$ to be this edge only, while otherwise we take a component $D^i \in \mathcal{C}$ with
$v^i,w \in N(D^i)$ (whose existence is promised by Theorem~\ref{thm:pmc}) and take as $P^i$ a shortest path from $v^i$ to $w$
with internal vertices in $C^i$. Note that $D^i \notin \{C^1,C^2\}$, since $w \in N(D^i)$. Furthermore, $v^{3-i} \notin D^i$,
as no component other than $C^1$ and $C^2$ can neighbour both $v^1$ and $v^2$. Consequently, $D^1 \neq D^2$, and the concatenation
of $P^1$ and $P^2$ gives the desired path from $v^1$ to $v^2$.
\end{proof}

By Lemma~\ref{lem:linked2}, the pair $\{C_1,C_2\}$ is linked unless $N(C_1) \cup N(C_2) = \pmc$.
However, if this is the case, by Theorem~\ref{thm:pmc} we have that both $N(C_1) \setminus N(C_2)$ 
and $N(C_2) \setminus N(C_1)$ are nonempty. By Lemma~\ref{lem:Cuniv}, there exists $i \in \{1,2\}$
and a vertex $v^i \in \pmc$ that is fully adjacent to $C_i$. Consequently, every efficient dominating
set consistent with $(X_0,Y)$ contains exactly one vertex of $C_i$: it needs to contain at least one
to dominate $Y \cap C_i$, but at most one since every vertex of $C_i$ dominates $v^i$.

We branch into $|Y \cap C_i|$ directions, guessing the vertex from $Y \cap C_i$ that belongs to the solution, and putting it
into $X_0$. Furthermore, in every branch we remove from $Y$ all vertices of $Y \cap C_i$.
In every subcase, $C_i$ is no longer an active component, but witnesses that every two other components that remain active
are linked: since $B \subseteq N(C_i)$,  we can always route a path between the desired endpoints through $C_i$.

By the above analysis and branching step, we can assume henceforth that any pair of active components is linked.

\subsubsection{Branching on Bad Vertices}

Partition $Y$ into $Y_1 = \{y \in Y: |N_B(y)| \geq |B|/16\}$ and $Y_2 = Y \setminus Y_1$.
Let $Y_1^\ast$ be the set of vertices $y \in Y_1$ for which the addition of $y$ to the solution (i.e., to $X_{0}$) and the subsequent exhaustive application of the Reduction Rule reduces $B$ to an empty set. Let $Y_1^\circ = Y_1 \setminus Y_1^\ast$.

If we knew that some vertex of $Y_1^\ast$ belongs to the solution, we could just guess it and the Reduction Rule
would reduce the set $B$ completely. In this section we focus on the analysis of the set of ``bad'' vertices $Y_1^\circ$,
showing that any such vertex also gives ground to a good branching --- but in a completely different fashion.

Let $y \in Y_1^\circ$. Assume that if we add $y$ to $X_0$ and exhaustively apply the Reduction Rule,
we shrink the set $Y$ to $Y^\circ$ and $B$ to $B^\circ \neq \emptyset$.
We claim the following:
\begin{lemma}\label{lem:bad1}
For every $z \in N[y] \cap Y$, it holds that $N_B(y) \subseteq N(z)$ or $B^\circ \subseteq N(z)$.
\end{lemma}
\begin{proof}
See the left panel of Fig.~\ref{fig:ed:bad1} for an illustration of the proof.
Fix a vertex $z$ as in the statement, and assume the contrary: there exist $p \in N_B(y) \setminus N(z)$
and $q \in B^\circ \setminus N(z)$. Note that $z \neq y$, as $p \in N_B(y)$.
Since $q \in B^\circ$, there exists at least two components of $\mathcal{C}$ that contain vertices of $N(q) \cap Y^\circ$.
Let $C_q$ be one of these components that is different from $C(y)$, and let $s \in N(z) \cap C_q \cap Y^\circ$.
Since Reduction Rule does not trigger on $q$ and $s$ after $y$ has been put into $X_0$, there exists $t \in (N(s) \cap Y^\circ) \setminus N(q)$; clearly, $t$ lies also in $C_q$.

Observe that $q,s,t$ induce a $P_3$, while $s$ and $t$ are not adjacent to $p \in N_B(y)$ as $s,t \in Y^\circ$. 
Furthermore, $p,y,z$ induce a $P_3$, while $y$ and $z$ are not adjacent to $q$.
Since $C_q$ and $C(y)$ are linked, we can connect $p$ and $q$ by an induced path avoiding $N[C_q] \cup N[C(y)]$, giving together
with the aforementioned $P_3$'s an induced path on at least six vertices, a contradiction.
\end{proof}

\begin{figure}[tb]
\begin{center}
\includegraphics{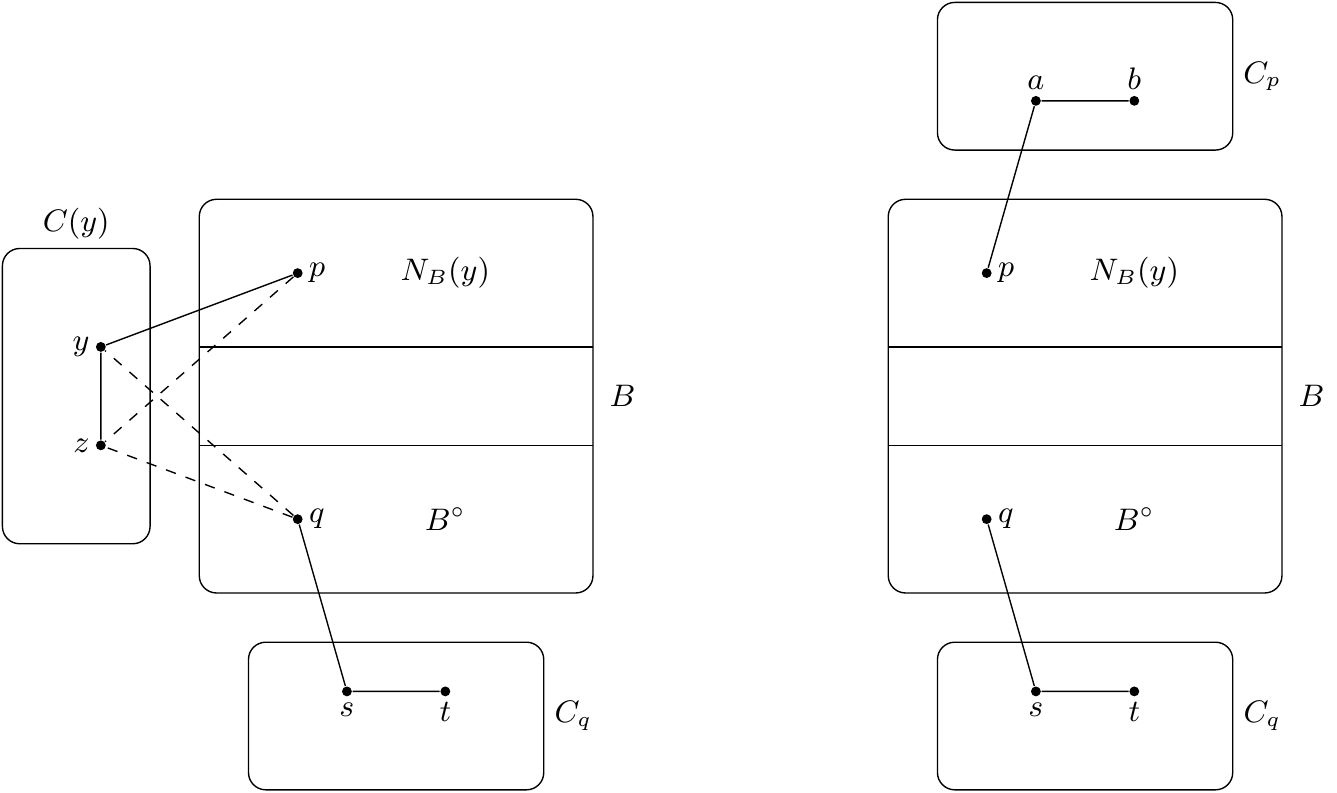}
\caption{Illustration of the proof of Lemma~\ref{lem:bad1} (left) and of Lemma~\ref{lem:bad2} (right).}
\label{fig:ed:bad1}
\end{center}
\end{figure}

Lemma~\ref{lem:bad1} allows us to branch into two directions, deciding whether the element of the sought efficient dominating set
that dominates the vertex $y$ also dominates the set $N_B(y)$ or the set $B^\circ$. 
That is, in the first subcase, we delete from $Y$ all vertices of $N(N_B(y)) \setminus C(y)$, while
in the second subcase, we delete from $Y$ all vertices of $N(B^\circ) \setminus C(y)$. 
We claim that in both subcases, after applying exhaustively the Reduction Rule, the size of $B$ decreased
at least by a multiplicative factor of $1-1/16$.

Clearly this is the case in the first subcase, as then $N_B(y)$ is removed from $B$ and $|N_B(y)| \geq |B|/16$ since
$y \in Y_1$. We claim the following:

\begin{lemma}\label{lem:bad2}
In the second subcase, the Reduction Rule also removes the entire set $N_B(y)$ from $B$.
\end{lemma}
\begin{proof}
See the right panel of Fig.~\ref{fig:ed:bad1} for an illustration of the proof.
Assume that this is not the case. Let $B',Y'$ be the reduced sets $B'$ and $Y'$ in the second subcase, and let $p \in N_B(y) \cap B'$.
Since $p \in B'$, there exist at least two components of $\mathcal{C}$ that contain vertices of $N(p) \cap Y'$; let $C_p$
be such component different than $C(y)$, and let $a \in C_p \cap N(p) \cap Y'$. Since the Reduction Rule does not trigger
on $p$ and $a$, given sets $B'$ and $Y'$, there exists $b \in (N(a) \cap Y') \setminus N(p)$; clearly also $b \in C_p$.

Consider now any $q \in B^\circ$. Since $a,b \in Y'$, we have that $a$ and $b$ are not adjacent to $q$.
Furthermore, since $q \in B^\circ$, there exist at least two components of $\mathcal{C}$ that contain vertices of $N(q) \cap Y^\circ$;
let $C_q$ be such a component different than $C_p$, and let $s \in C_p \cap N(q) \cap Y^\circ$. 
Since the Reduction Rule does not trigger on $q$ and $s$ given sets $B^\circ$ and $Y^\circ$, there exists
$t \in (N(s) \cap Y^\circ) \setminus N(q)$; clearly also $t \in C_q$. Furthermore,
$p$ is not adjacent to neither $s$ nor $t$, as $s,t \in Y^\circ$ and $p \in N_B(y)$.

Consequently, the vertices $p,a,b,q,s,t$, together with a path between $p$ and $q$ promised by the fact that $C_p$ and $C_q$
are linked, induce a path on at least six vertices, a contradiction.
\end{proof}

We infer that in both subcases at least a constant fraction of the set $B$ has been reduced.
Consequently, in the branching tree, every leaf-to-root path contains only $\Oh(\log n)$ nodes with a branching
described in this section.

\subsubsection{Final Branch}

We are left with cases $(X_0,Y)$ when $Y_1^\circ = \emptyset$.
Consider the following natural branch: we guess whether there exists an element of the solution in $Y_1^\ast$ or not.
That is, in one branch we remove $Y_1^\ast = Y_1$ from $Y$. In the second branch, we immediately branch
again into $|Y_1^\ast|$ directions, picking a vertex $y \in Y_1^\ast$ and putting it into $X_0$.
By the definition of $Y_1^\ast$, in the latter subcases $B$ is reduced to an empty set, and branching terminates.
Our main claim is that in the first branch, the size of $B$ shrinks by at least a half.

\begin{lemma}\label{lem:final1}
In the first branch, if $B'$ and $Y'$ are the sets $B$ and $Y$ after exhaustive application of the Reduction Rule,
then $|B'| \leq |B|/2$.
\end{lemma}
\begin{proof}
Assume the contrary. Since $Y' \subseteq Y_2$, we have that for every $y \in Y'$
it holds that
$$|N_{B'}(y)| \leq |N_B(y)| \leq |B|/16 < |B'|/8.$$
For every $p \in B'$, pick two components $C^1_p$ and $C^2_p$ that contain a vertex of $N(p) \cap Y'$.
Furthermore, for every $i=1,2$, pick a vertex $v_p^i \in N(p) \cap Y' \cap C^i_p$ and a vertex
$w_p^i \in (N(v_p^i) \cap Y') \setminus N(p)$; the existence of the latter is guaranteed by the fact
that the Reduction Rule does not trigger on $p$ and $v_p^i$, given the sets $B'$ and $Y'$.
Clearly, $p,v_p^i,w_p^i$ induce a $P_3$ in $G$.

Consider the following random experiment: choose two vertices $p,q \in B'$ uniformly independently at random.
Since the choice of $p$ and $q$ is independent, while all vertices $v^i_p,w^i_p,v^i_q,w^i_q$ belong to $Y'$,
the probability that $q$ is adjacent to $v_p^i$ is less than $1/8$. Consequently, with positive probability
$q$ is fully anti-adjacent to $\{v^1_p,w^1_p,v^2_p,w^2_p\}$, while
$p$ is fully anti-adjacent to $\{v^1_q,w^1_q,v^2_q,w^2_q\}$.

Let $p,q$ be a pair for which the aforementioned event happens. By potentially swapping the top indices,
we may assume $C^1_p \neq C^1_q$. However, then $p,v^1_p,w^1_p, q, v^1_q,w^1_q$, together
with an induced path between $p$ and $q$ whose existence is promised by the fact that $C^1_p$ and $C^1_q$ are linked,
gives an induced path in $G$ on at least six vertices, a contradiction.
\end{proof}

By Lemma~\ref{lem:final1}, on every leaf-to-root path a branching node described in this section
may appear only $\Oh(\log n)$ times. This finishes the description of the algorithm, and concludes
the proof of Theorem~\ref{thm:effdom-states}.

\subsection{The Actual Algorithm}

As described in the beginning of the section, the actual algorithm for \EDname{} is a standard dynamic programming algorithm,
using Theorem~\ref{thm:effdom-states} as the source of its state space.
For sake of analysis, we fix $X_0$ to be a maximum weight efficient dominating set in $G$ (if such a set exists).

We first pick any minimal triangulation $\cG$ of $G$ (see~\cite{pinar-survey} for algorithms finding such a triangulation), and compute
its clique tree, which is at the same
time a tree decomposition of $G$. In other words, we compute a tree decomposition $(T,\beta)$ of $G$, where for
every node $t \in V(T)$ the bag $\beta(t)$ is a potential maximal clique of $G$. 

We root $T$ at an arbitrary vertex $r$, and for a node $t$ we denote by $\gamma(t)$ the union of all bags $\beta(s)$, where $s$ ranges over all descendants of $t$ in the tree $T$.
Note that the properties of a tree decomposition ensure that every connected component of $G-\beta(t)$ is either completely
contained in or completely disjoint from $\gamma(t)$.

For every $t \in V(T)$, we invoke Theorem~\ref{thm:effdom-states}, obtaining a family $\mathcal{S}_t$.

For a node $t$, a set $Y \subseteq \gamma(t)$ is called a \emph{partial solution}
if $N[u] \cap N[v] = \emptyset$ for every distinct $u,v \in X$ and $N[X]$ contains $\gamma(t) \setminus \beta(t)$.
Clearly, if $X$ is an efficient dominating set in $G$, then $X \cap \gamma(t)$ is a partial solution.
A partial solution $Y$ is \emph{consistent} with a state $f \in \mathcal{S}_t$ if 
$Y \cap \beta(t) = f^{-1}(\bot)$, every vertex $v \in \beta(t) \cap N(Y)$ is dominated by an element of $Y$
in $f(v)$, and for every vertex $v \in \beta(t) \setminus N[Y]$ the component $f(v)$ is disjoint from $\gamma(t)$.

Our goal is to compute, in bottom-up fashion, for every node $t \in V(T)$ and every state $f \in \mathcal{S}_t$
a partial solution $Y(t,f)$ consistent with $f$ (or $Y(t,f) = \bot$, meaning that no such set has been found),
  with the following property: if $X_0$ exists and
$f$ is consistent with $X_0 \cap \gamma(t)$,
  then $Y(t,f)$ exists and has weight at least the weight of $X_0 \cap \gamma(t)$.
  Note that Theorem~\ref{thm:effdom-states} ensures that if $X_0$ exists then 
  for every node $t$ there exists a state $f^t_0$ consistent with $X_0$,
  and thus also consistent with partial solution $X_0 \cap \gamma(t)$.
Consequently, if $X_0$ exists, then $Y(r,f_0^r)$ is a maximum weight efficient dominating set in $G$.

It remains to describe the computation for fixed values $t$ and $f$ and prove the aforementioned property.
For every child $t'$ of $t$, and every $f' \in \mathcal{S}_{t'}$, we say that the set $Y' := Y(t',f')$ is
\emph{partially consistent} with $f$ if $Y' \cap \beta(t) \cap \beta(t') = f^{-1}(\bot) \cap \beta(t')$,
every vertex $v \in \beta(t) \cap N(Y')$ is dominated by an element of $Y'$ in $f(v)$,
and for every vertex $v \in \beta(t) \setminus N[Y']$ the component $f(v)$ is disjoint from $\gamma(t')$ or equals $\pmc$.
For every child $t'$ of $t$ we compute a maximum weight set $Y_{t'}$
among all sets $Y(t',f')$ for $f' \in \mathcal{S}_{t'}$ that are partially consistent with $f$;
we terminate the computation and set $Y(t,f) = \bot$ if for some child $t'$ the set $Y_{t'}$ does not exist
(i.e., we picked the maximum over an empty set).
A direct check shows that if all sets $Y_{t'}$ has been computed, then
the union of all sets $Y_{t'}$ is a partial solution consistent with $f$,
and we pick it as $Y(t,f)$.

Consider now the state $f_0^t$, and assume that for every child $t'$ of $t$,
the set $Y(t',f_0^{t'})$ exists and has weight at least at the weight of $X_0 \cap \gamma(t')$.
Observe that $Y(t',f_0^{t'})$ is also partially consistent with $f_0^t$.
A direct check from the definition of consistency shows that 
for any $Y(t',f')$ partially consistent with $f_0^t$, the set
$X_0' := (X_0 \setminus \gamma(t')) \cup Y(t',f')$ is also an efficient dominating set.
Since $Y_{t'}$ is chosen to be a set $Y(t',f')$ of maximum weight that is partially consistent with $f_0^t$, and $Y(t',f_0^{t'})$ is one of the candidates,
$X_0'$ is a maximum weight efficient dominating set.
By repeating this replacement argument for every child $t'$ of $t$,
we infer that the computed value $Y(t,f_0^t)$ has weight at least the
weight of $X_0 \cap \gamma(t)$.

Since the computations are polynomial in the size of $G$ and the sizes of the families $\mathcal{S}_t$, using Theorem~\ref{thm:effdom-states} we conclude the proof of Theorem~\ref{thm:effdom}.

\section{Conclusions}\label{sec:conclusion}
We have shown a quasipolynomial-time algorithm for \ISname{} and a polynomial-time algorithm for \EDname{} in $P_6$-free graphs.
Our algorithms rely on a detailed analysis of the interactions between minimal separators, potential maximal cliques, and vertex neighborhoods in $P_6$-free graphs.

In light of these developments, a few open questions seem natural for the \ISname{} problem. First, can \ISname{} on $P_6$-free graphs be solved in polynomial time? 
Second, can the quasipolynomial-time algorithm be generalized to $P_7$-free graphs? Theorems~\ref{thm:hit-sep} and~\ref{thm:hit-pmc} work for $P_7$-free graphs, but
Theorem~\ref{thm:hit-nuke} does not, as can be seen on the following example. Consider the graph $G$ consisting of $k+1$ cliques on $k$ vertices each, denoted $A,C_1,C_2,\ldots,C_k$,
with a vertex $c_i$ distinguished in every clique $C_i$ and made adjacent to a private vertex $a_i \in A$ (see the left panel of Fig.~\ref{fig:counter}). 
$G$ contains many $P_6$'s with middle two vertices in $A$, but no $P_7$. The set $X = \{c_1,c_2,\ldots,c_k\}$ is a nuke in $G$, but no vertex of $G$ is adjacent to more than one vertex of $X$.
Furthermore, if one adds a new vertex $y$ to $G$ that is adjacent to $X$, then $X \cup \{y\}$ becomes a potential maximal clique. Recall that the algorithm for \ISname{} works by picking a central PMC
as a pivot nuke, and branching on vertices adjacent to a constant fraction of the pivot nuke. Hence, with a similar approach on $P_7$-free graphs the algorithm may end up with such a seemingly useless nuke as $X$ in $G$.

We remark that Theorems~\ref{thm:hit-sep} and~\ref{thm:hit-pmc} also do not seem to generalize to less restrictive graph classes. 
Consider a graph $G$ consisting of $k+2$ cliques on $k$ vertices each, denoted $A,B,S_1,S_2,\ldots,S_k$, with every $S_i$ adjacent to a private vertex $a_i \in A$ and $b_i \in B$.
The set $S := \bigcup_{i=1}^k S_i$ is a minimal separator in $G$ of size $k^2$ with $A$ and $B$ as full components, yet no vertex of $G$ contains more than $k$ vertices of $S$ in its neighborhood.
Furthermore, although $G$ contains many $P_7$'s with endpoints and middle vertex in $S$, it does not contain a $P_8$ nor an $E$-graph (a $P_5$ with an additional degree-$1$ vertex attached to the middle vertex of the path).

\begin{figure}[tb]
\begin{center}
\includegraphics{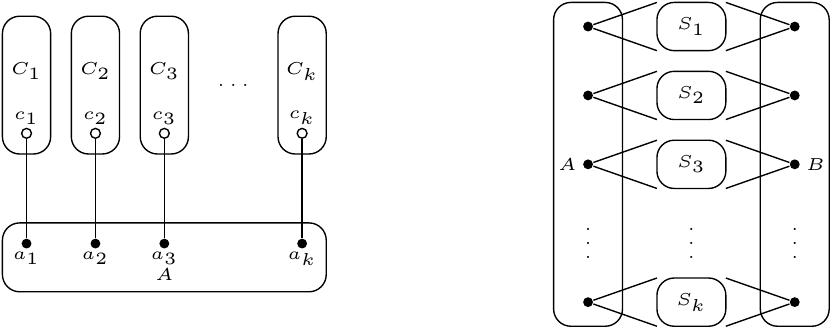}
\caption{Counterexamples to generalizations of Theorem~\ref{thm:hit-nuke} to $P_7$-free graphs (left panel, the nuke vertices are white)
and of Theorem~\ref{thm:hit-sep} to $P_8$-free graphs (right panel). Every rectangle denotes a clique on $k$ vertices.}
\label{fig:counter}
\end{center}
\end{figure}

Finally, we remark that although the polynomial-time algorithm for \EDname{} on $P_6$-free graphs seems to close a research direction (as the problem is NP-hard on $P_7$-free chordal graphs),
it would be interesting to see if one can obtain the same end result using the approach of~\cite{ab1,ab2}, that is, by either obtaining a polynomial-time algorithm for \ISname{} in hole-free graphs or showing that the square of a $P_6$-free graph having an efficient dominating set is perfect.

\section*{Acknowledgements}

The second author acknowledges discussions with Krzysztof Choroma\'{n}ski, Dvir Falik, Anita Liebenau, and Viresh Patel
on the usage of minimal separators and potential maximal cliques in study of the Erd\H{o}s-Hajnal conjecture; 
in particular, the current proofs of Theorems~\ref{thm:hit-sep} and~\ref{thm:hit-pmc},
(that replaced our previous proofs more heavy on case analysis) were partially inspired
by some quantitative proof attempts in the study of subcases of the Erd\H{o}s-Hajnal conjecture.

\bibliographystyle{abbrv}
\bibliography{../p6}

\end{document}